\documentclass[letterpaper, 10 pt, journal]{ieeetran}
\usepackage{cite}
\usepackage{amsmath,amssymb,amsfonts}
\usepackage{graphicx}
\usepackage{textcomp}
\usepackage{algpseudocode} 
\usepackage{amsthm}
\usepackage{logicproof}
\usepackage[flushleft]{threeparttable}
\usepackage{amssymb}
\usepackage{lipsum}
\usepackage{etoolbox}
\usepackage{enumitem}
\usepackage{float}
\usepackage{booktabs}
\newtheorem{definition}{Definition}
\usepackage{multirow}
\newtheorem{lemma}{Lemma}

\def\BibTeX{{\rm B\kern-.05em{\sc i\kern-.025em b}\kern-.08em
    T\kern-.1667em\lower.7ex\hbox{E}\kern-.125emX}}

\IEEEoverridecommandlockouts                              

\usepackage{graphics} 
\usepackage{epsfig} 
\usepackage{amsmath} 
\usepackage{amssymb}  

\usepackage{url}
\usepackage[ruled, vlined, linesnumbered]{algorithm2e}
\usepackage{verbatim} 
\usepackage{soul, color}
\usepackage{lmodern}
\usepackage{fancyhdr}
\usepackage[utf8]{inputenc}
\usepackage{fourier} 
\usepackage{array}
\usepackage{makecell}

\SetNlSty{large}{}{:}

\title{\LARGE CEV Framework: A Central Bank Digital Currency Evaluation and Verification Framework With a Focus on Consensus Algorithms and Operating Architectures}

\author{Si Yuan JIN, Yong Xia\IEEEauthorrefmark{1}\thanks{*Corresponding Author: Yong Xia (yong.xia@hsbc.com)
\\ 
This work has been accepted by the IEEE Access on 14 June 2022. Digital Object Identifier 10.1109/ACCESS.2022.3183092.
}
\\ HSBC Laboratory, Guangzhou 510440, China  \\}

\begin{document}

\thispagestyle{plain}
\pagestyle{plain}

\maketitle
\begin{abstract}
We propose a Central Bank Digital Currency Evaluation and Verification (CEV) Framework
for recommending and verifying technical solutions in the central bank digital currency (CBDC) system.
We demonstrate two sub-frameworks: an evaluation sub-framework that provides consensus algorithm and
operating architecture solutions and a verification sub-framework that validates the proposed solutions.
Our framework offers a universal CBDC solution that is compatible with different national economic
and regulatory regimes. The evaluation sub-framework generates customized solutions by splitting the
consensus algorithms into several components and analyzing their impacts on CBDC systems. CBDC design
involves a trade-off between system features - the consensus algorithm cannot achieve all system features
simultaneously. However, we also improve the operating architectures to compensate for the weak system
features. The verification sub-framework helps verify our proposed solution through empirical experiments
and formal proof. Our framework offers CBDC designers the flexibility to iteratively tune the trade-off
between CBDC system features for the desired solution. To the best of our knowledge, we are the first to
propose a framework to recommend and verify CBDC technical solutions.

\end{abstract}

\maketitle

\section{Introduction}
\label{sec:introduction}
The recent development in cryptography and distributed ledger technology (DLT) has seen a new form of currency known as Central Bank Digital Currency (CBDC) \cite{Theriseofdigitalcurrency}. More than 85\% of central banks worldwide have already started CBDC research \cite{b2, CBDCSurvey}. However, most current CBDC research works come from central banks, while only a few scientific papers discuss the CBDC-related research problems. For example, the papers \cite{BlockchainCBDCIEEE,blockchaincentralbank} used blockchain networks to provide CBDC services and propose a new consensus algorithm to satisfy CBDC technical features. However, these scientific papers only discuss some specific scenarios and have limitations to extending to a different scenario. 

The overall operating architecture\cite{BISFourModel} and consensus algorithms are core parts of CBDC technical solutions. The overall operating architecture defines different CBDC networks. Consensus algorithms define how the specific CBDC network functions and impacts many CBDC technical features. CBDC technical features\cite{groupCBDC} measure the technical focuses from different central banks. Here the consensus algorithm does not have to be applied in a blockchain network. Any network can apply a consensus algorithm to form data consistency. For example, China\cite{Yao} did not use blockchain to design its CBDC prototype, but we still consider the way to form data consistency as one kind of consensus algorithm. Due to diverse national conditions, central banks need different consensus algorithms and operating architectures to satisfy their CBDC technical features.

\subsection{Our Contribution}
Our paper reviewed previous CBDC solutions and proposed a framework that provides overall operating architecture and customized consensus algorithms to satisfy different CBDC technical features. Section II shows three CBDC technical features: performance, security and privacy. Compared with previous works, we have the following contributions:
\begin{enumerate}
    \item We propose a framework to recommend and verify CBDC related technical solutions in consensus algorithms and operating architectures.
    \item We are the first to split consensus algorithms into different components, significantly improving the efficiency to design customized consensus algorithms.
    \item We improve the CBDC operating architecture to solve the business secrecy issue.
\end{enumerate}

Specifically, we propose an evaluation sub-framework that provides holistic CBDC solutions covering CBDC technical features in Section III-A and build a verification sub-framework to verify the feasibility and rationality of proposed solutions in Section III-B. Finally, we integrate both sub-frameworks into one framework, called CEV Framework. 

The evaluation sub-framework involves the consensus algorithm and operating architecture. Consensus algorithm works for forming data consistency among participants \cite{RAFT}. It impacts many CBDC technical features directly, including performance, privacy and security. For example, the paper \cite{hybrid} studied how blockchain empowers CBDC and proposed a new consensus algorithm to improve CBDC performance. However, consensus algorithms have a highly complex impact on CBDC technical features. In order to better analyze consensus algorithms, we split them into different components so that we can derive the impacts of each component on CBDC technical features. 

A trade-off \cite{groupCBDC} between CBDC technical features exists in implementing consensus algorithms, which means we can not achieve all simultaneously. However, we can improve the operating architecture to compensate for weak CBDC technical features. Our paper uses a new operating architecture to solve the business secrecy issue (details in Section II-B).

The verification sub-framework can guide CBDC designers to verify proposed solutions. CBDC designers need to build a mathematical model for the solution and verify whether it can meet initial expectations on diverse CBDC technical features, like high performance. If they need further adjusting preference on three CBDC technical features, they can go back to the evaluation sub-framework to adjust the previous solution again.

We then introduce a CBDC scenario to demonstrate the CEV framework in Section IV. We used the evaluation sub-framework to propose customized consensus algorithms and followed the verification sub-framework to build a model and verify related CBDC technical features. We used empirical experiments to test performance and leveraged formal proofs to verify security and privacy. Our framework can give a clear guide to satisfy CBDC technical features for CBDC designers despite different national economic and regulatory conditions (details in Section III.A.1).

\subsection{Paper Structure}
The remainder of this paper is organized as follows. In Section II, we present the background of the research and three CBDC technical features. Section III introduces the CEV framework, including an evaluation sub-framework and a verification sub-framework. Section IV gives an example of leveraging the framework to develop a solution and verify it. Finally, we conclude the paper in Section V.

\section{Background}
\subsection{Blockchain and Consensus Algorithm}
Blockchain has shown many benefits among current CBDC projects worldwide \cite{Currencies2020, CentreforEconomicPolicyResearchCEPR2019, BlockchainCBDCACM, BlockchainCBDCIEEE, blockchainframework}. For example, in cross-border business, peer-to-peer payment could save liquidity and improve efficiency \cite{BoT2020}. Research topics about blockchain appear one after another, especially in the performance part\cite{cachin2016architecture, Brown2016,eyal2016bitcoin,poon2016bitcoin}, since the performance of blockchain, for example, bitcoin \cite{nakamoto2012bitcoin}, cannot meet today's commercial needs. Besides performance, other CBDC technical features are also undergoing research, like resilience, security and privacy. 

Consensus algorithms play roles in many blockchains. Fabric \cite{cachin2016architecture} used Practical Byzantine Fault Tolerance (PBFT) \cite{castro1999practical} which provides fault tolerance while sacrificing part of performance. Corda blockchain protocol aims to satisfy finance and regulatory requirements\cite{Brown2016}. Transactions in the Corda platform are recorded only by participants rather than the entire network, which provides high performance and protects privacy while sacrificing part of security. 

\subsection{Tiered Architecture and Business Secrecy Issue}
Most CBDC pilots have shown a tiered architecture \cite{BISFourModel} which plays roles in many CBDC projects, including China's E-CNY \cite{designchoice}, Sweden's E-Krona \cite{sweden3}. Figure \ref{fig:ConsensusNetworks} shows a typical tiered CBDC architecture. We define a consensus network as the network which consensus algorithms run.

\begin{figure}[h]
	\centering
	\includegraphics[width=0.50\textwidth]{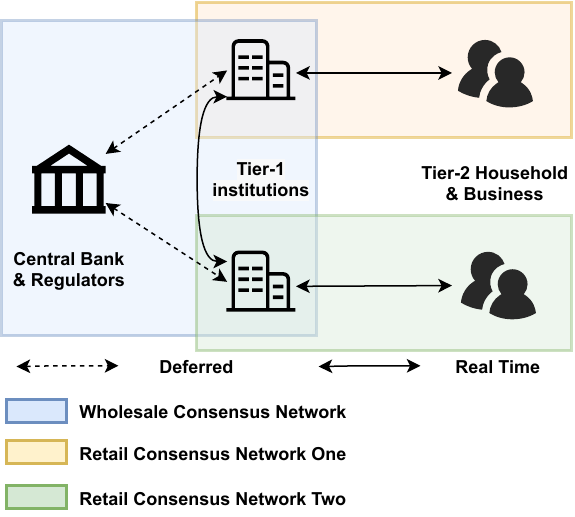}
	\caption{\textbf{A two-tier architecture that consensus algorithms run separately in different consensus networks. The wholesale consensus network involves central banks and tier-1 institutions and handles wholesale transactions between tier-1 institutions and central banks; the retail consensus networks involve retail clients and tier-1 institutions and handle retail transactions between tier-2 households \& business and tier-1 institutions.}\label{fig:ConsensusNetworks}}
\end{figure}

In a two-tier CBDC architecture, tier-1 institutions directly connect to the central bank (tier-0), and tier-2 institutions directly connect to tier-1 institutions. Tier-1 institutions take the responsibility of distribution\footnote{Distribution means that an institution helps the central bank issue CBDC and manage CBDC authentication work.} in a two-tier CBDC architecture. In most CBDC projects, commercial banks become tier-1 institutions. However, it is impossible to let all commercial institutions become tier-1 and responsible for CBDC distribution and circulation\footnote{Circulation means that an institution provides CBDC-related transfer services.} because the central bank can not afford too many banks to connect simultaneously. At the same time, it has a potential single point (central bank) failure risk and performance bottleneck. Therefore, in most countries, the most influential banks usually became tier-1 institutions.

The more commercial institutions circulate CBDC, the more areas CBDC services cover. Tier-2 institutions have to connect to tier-1 institutions to provide CBDC services. The two-tier model can mitigate business secrecy-related concerns if tier-2 institutions and tier-1 institutions are non-competitors. However, tier-2 institutions and tier-1 institutions are mostly competitors, and they are reluctant to provide transaction and customer information to their competitors. For example, if all tier-1 institutions are banks, tier-2 banks are worried that tier-1 banks monopolize their customer data. The business secrecy issue makes two-tier architecture hard to implement in a CBDC system. 

Current technical solutions to keep business secrecy, such as homomorphic encryption \cite{maturesolutionprivacy}, however, could not satisfy CBDC technical features because it influences performance a lot. Therefore, we propose new operating architectures to improve CBDC business secrecy (details in Section III-A.2).

\subsection{CBDC technical features}
CBDC technical features \cite{groupCBDC} measure CBDC-related considerations for designers and regulators. CBDC white papers presented many differences between jurisdictions regarding national conditions, and central banks focus on different CBDC technical features. For example, Singapore's Ubin \cite{b9}, and Canada's Jasper \cite{b7} focuses on transaction settlement between different countries; China's E-CNY \cite{designchoice} emphasizes the volume of transactions per second in retail transactions. CBDC designers across different jurisdictions have varying approaches to satisfy CBDC technical features. 

We rest on the previous research \cite{b3,b6,Security,Boar2020, b7,b9,b10,b12,b13,retailCBDCtechnology,retailimf,b15} and extracted following CBDC technical features. To avoid ambiguity, we use specific sub-categories to define the following categories.
 
\subsubsection{Performance}
Blockchain has many benefits and has been widely used in the wholesale CBDC, but it seldom appears in retail CBDC projects\cite{m-cbdc}. One key factor is that its weak scalability cannot meet high performance.
In CBDC scenarios, millions, even billions of customers may use CBDC, which requires a high performance to handle billions of requests. Therefore, we consider the following features to measure performance.

\begin{enumerate}
\item User Scalability: the cost of adding a new customer to a CBDC system. 
\item Network Scalability: the capability to handle larger transaction volumes per second(TPS).
\item Latency: the time to complete one transaction.
\end{enumerate}
We used empirical experiments to examine performance in the verification sub-framework and gave an example in Section IV-B.2.

\subsubsection{Security}
Security in a distributed system involves various aspects, including cryptography, secure channels, key management, prevention of double-spending attacks \cite{distributedcomputingbook}. The prevention of double-spending is one of the basic requirements in a CBDC system and maintains financial stability and reliability. Central banks usually put security as the priority. As one kind of new form currency, CBDC has many security risks. 

\begin{enumerate}
\item Cyber-Security: capability of protecting against outside attacks, especially double-spending attacks.
\item Resilience: capability of protecting against hardware issues, power or network outages, or cloud service interruption\cite{retailimf}.
\end{enumerate}
We used formal proof to verify potential security threats in the verification sub-framework and gave an example in Section IV-B.4.
\subsubsection{Privacy}
We divide privacy into two aspects, customer privacy and business secrecy\cite{retailimf}.
\begin{enumerate}
\item Customer privacy protects customer data against others.
\item Business secrecy prevents business data from leaking to business competitors. second(TPS).
\end{enumerate}
Our evaluation sub-framework improved the current operating architecture to protect business secrecy in Section III-A.2 since the consensus algorithm can not simultaneously achieve all CBDC technical features. We used formal proof to verify privacy protection in the verification sub-framework and gave an example of Section IV-B.3.

\subsubsection{Others}
Other CBDC technical features do not conflict with the above ones. Examples include governance \cite{retailimf}, functionality, interoperability and offline payments. We believe these requirements can be met or solved independently in the current financial system. For example, offline payment does not conflict with the above CBDC technical features. So we ignored them in the following parts. In future, we can involve more CBDC technical features in the CEV framework if needed.

\section{CEV Framework}
\begin{figure*}[h]
	\centering
	\includegraphics[width=1.0\textwidth]{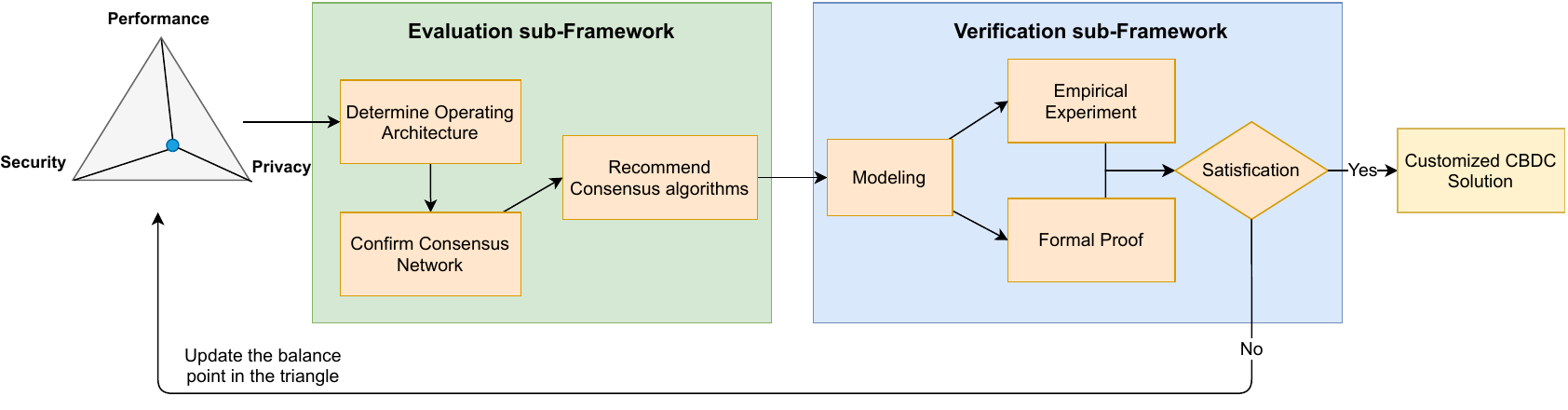}
	\caption{\textbf{A closed-loop workflow of CEV Framework for CBDC designers.} \label{fig:General Workflow}}
\end{figure*}

The CEV (CBDC Evaluation and Verification) framework includes two sub-frameworks: an evaluation sub-framework that provides CBDC solutions and a verification sub-framework that proves the feasibility and rationality of recommended solutions. 

Figure \ref{fig:General Workflow} shows the working procedures of the CEV framework. First, CBDC designers determine preferable CBDC technical features according to their economic and regulatory conditions. Then the evaluation sub-framework firstly determine the operating architecture to ensure how many consensus networks and the relationship between them. Secondly, the evaluation sub-framework recommends consensus algorithms in different consensus networks. Afterwards, The verification sub-framework can guide CBDC designers to build a theoretical model for the solution and carry out experiments and proofs to verify whether the solution meets the original CBDC technical features. Finally, suppose they are not satisfied with the solution. In that case, they can go back to adjust their preference on CBDC technical features, leverage the evaluation sub-framework to update their solutions, and use the verification sub-framework to check the proposed solutions again.

\subsection{Evaluation Sub-framework}
The evaluation sub-framework includes two parts: the consensus algorithm part splits consensus algorithms into different components, and the operating architecture part introduces the overall architecture that consensus algorithms run. Next, we introduce how we recommend consensus algorithms and improve the overall operating architecture.

\subsubsection{Consensus algorithm}
Consensus networks have diverse implementations of the consensus algorithm. For example, in the operating architecture (figure \ref{fig:ConsensusNetworks}), the central bank can control the wholesale balance of issued CBDC rather than recording every retail transaction to avoid double-spending \cite{distributedcomputingbook}. The central bank is only responsible for issuance and redemption transactions. If any issue exists in retail transactions, corresponding tier-1 institutions should be accountable. On the other side, the central bank records every wholesale transaction to keep it safe. Then both the wholesale and retail transactions are safe from double-spending from the perspective of central banks.

\begin{figure*}[h]
	\centering
	\includegraphics[width=1.00\textwidth]{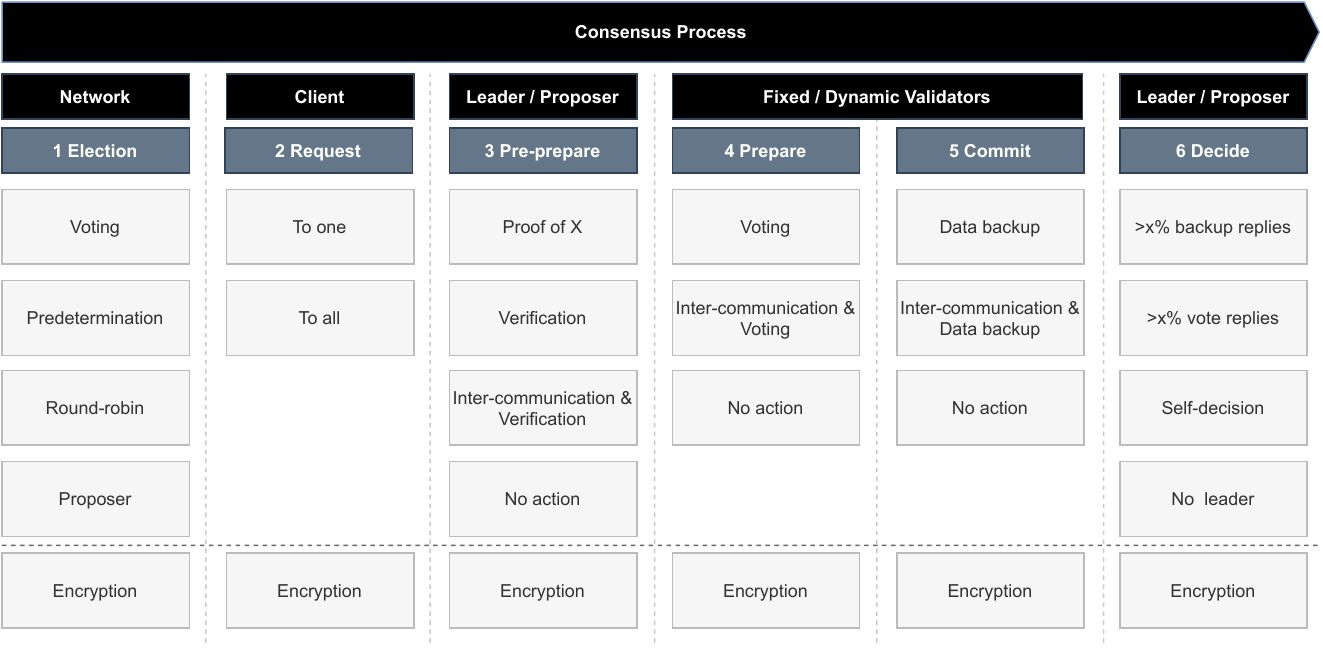}
	\caption{\textbf{Consensus Algorithm Process Map. A CBDC client sends a request to the CBDC system, and the system processes the request until reaching an agreement inside the network. We split the process into several parts to better analyze each one.} \label{fig:ConsensusModularization}}
\end{figure*}
\begin{table*}[h]
\centering
\caption{\textbf{Impact of Individual Consensus Components. The first column is consensus components. The first row is CBDC-related CBDC technical features.}}
\label{Table:Impact}
\begin{tabular}{cc|ccc|cc|cc}
\hline
\multicolumn{2}{c|}{\multirow{2}{*}{\textbf{Modules}}}                                                                                                                                         & \multicolumn{3}{c|}{\textbf{Performance}}                                                                                                                                                   & \multicolumn{2}{c|}{\textbf{Security}}                               & \multicolumn{2}{c}{\textbf{Privacy}}                                                                                                          \\ \cline{3-9} 
\multicolumn{2}{c|}{}                                                                                                                                                                          & \multicolumn{1}{c|}{\begin{tabular}[c]{@{}c@{}}\textbf{User}\\ \textbf{Scalability}\end{tabular}} & \multicolumn{1}{c|}{\begin{tabular}[c]{@{}c@{}}\textbf{Network}\\ \textbf{Scalability}\end{tabular}} & \textbf{Latency}     & \multicolumn{1}{c|}{\textbf{Resilience}}   & \textbf{Cyber-Security} & \multicolumn{1}{c|}{\begin{tabular}[c]{@{}c@{}}\textbf{Customer}\\ \textbf{Privacy}\end{tabular}} & \begin{tabular}[c]{@{}c@{}}\textbf{Business}\\  \textbf{Secrecy}\end{tabular} \\ \hline
\multicolumn{1}{c|}{\multirow{5}{*}{\begin{tabular}[c]{@{}c@{}}1 Leader Election\\  / Proposer\end{tabular}}} & Voting                                                                         & \multicolumn{1}{c|}{\multirow{3}{*}{TBA}}                                       & \multicolumn{1}{c|}{Medium}                                                        & Medium               & \multicolumn{1}{c|}{\multirow{11}{*}{TBA}} & High                    & \multicolumn{1}{c|}{\multirow{2}{*}{TBA}}                                       & TBA                                                         \\ \cline{2-2} \cline{4-5} \cline{7-7} \cline{9-9} 
\multicolumn{1}{c|}{}                                                                                         & Predetermination                                                               & \multicolumn{1}{c|}{}                                                           & \multicolumn{1}{c|}{High}                                                          & High                 & \multicolumn{1}{c|}{}                      & \multirow{8}{*}{TBA}    & \multicolumn{1}{c|}{}                                                           & High                                                        \\ \cline{2-2} \cline{4-5} \cline{8-9} 
\multicolumn{1}{c|}{}                                                                                         & Round-robin                                                                    & \multicolumn{1}{c|}{}                                                           & \multicolumn{1}{c|}{High}                                                          & High                 & \multicolumn{1}{c|}{}                      &                         & \multicolumn{1}{c|}{Low}                                                        & Medium                                                      \\ \cline{2-5} \cline{8-9} 
\multicolumn{1}{c|}{}                                                                                         & Proposer                                                                       & \multicolumn{1}{c|}{High}                                                       & \multicolumn{1}{c|}{High}                                                          & Low                  & \multicolumn{1}{c|}{}                      &                         & \multicolumn{1}{c|}{\multirow{2}{*}{TBA}}                                       & \multirow{7}{*}{TBA}                                        \\ \cline{2-5}
\multicolumn{1}{c|}{}                                                                                         & Encryption                                                                     & \multicolumn{1}{c|}{TBA}                                                        & \multicolumn{1}{c|}{TBA}                                                           & Low                  & \multicolumn{1}{c|}{}                      &                         & \multicolumn{1}{c|}{}                                                           &                                                             \\ \cline{1-5} \cline{8-8}
\multicolumn{1}{c|}{\multirow{3}{*}{2 Request}}                                                               & To one                                                                         & \multicolumn{1}{c|}{High}                                                       & \multicolumn{1}{c|}{High}                                                          & Low                  & \multicolumn{1}{c|}{}                      &                         & \multicolumn{1}{c|}{Medium}                                                     &                                                             \\ \cline{2-5} \cline{8-8}
\multicolumn{1}{c|}{}                                                                                         & To all                                                                         & \multicolumn{1}{c|}{Low}                                                        & \multicolumn{1}{c|}{Medium}                                                        & High                 & \multicolumn{1}{c|}{}                      &                         & \multicolumn{1}{c|}{Low}                                                        &                                                             \\ \cline{2-5} \cline{8-8}
\multicolumn{1}{c|}{}                                                                                         & Encryption                                                                     & \multicolumn{1}{c|}{\multirow{2}{*}{TBA}}                                       & \multicolumn{1}{c|}{\multirow{2}{*}{TBA}}                                          & Low                  & \multicolumn{1}{c|}{}                      &                         & \multicolumn{1}{c|}{\multirow{3}{*}{TBA}}                                       &                                                             \\ \cline{1-2} \cline{5-5}
\multicolumn{1}{c|}{\multirow{5}{*}{3 Pre-Prepare}}                                                           & Proof of X                                                                     & \multicolumn{1}{c|}{}                                                           & \multicolumn{1}{c|}{}                                                              & TBA                  & \multicolumn{1}{c|}{}                      &                         & \multicolumn{1}{c|}{}                                                           &                                                             \\ \cline{2-5} \cline{7-7}
\multicolumn{1}{c|}{}                                                                                         & Verification                                                                   & \multicolumn{1}{c|}{High}                                                       & \multicolumn{1}{c|}{High}                                                          & High                 & \multicolumn{1}{c|}{}                      & High                    & \multicolumn{1}{c|}{}                                                           &                                                             \\ \cline{2-5} \cline{7-9} 
\multicolumn{1}{c|}{}                                                                                         & \begin{tabular}[c]{@{}c@{}}Inter-communication\\  \& Verification\end{tabular} & \multicolumn{1}{c|}{Low}                                                        & \multicolumn{1}{c|}{Low}                                                           & Low                  & \multicolumn{1}{c|}{}                      & High                    & \multicolumn{1}{c|}{Low}                                                        & Low                                                         \\ \cline{2-9} 
\multicolumn{1}{c|}{}                                                                                         & No action                                                                      & \multicolumn{1}{c|}{High}                                                       & \multicolumn{1}{c|}{High}                                                          & High                 & \multicolumn{1}{c|}{Medium}                & Low                     & \multicolumn{1}{c|}{Medium}                                                     & Medium                                                      \\ \cline{2-9} 
\multicolumn{1}{c|}{}                                                                                         & Encryption                                                                     & \multicolumn{2}{c|}{TBA}                                                                                                                                             & Low                  & \multicolumn{1}{c|}{\multirow{5}{*}{TBA}}  & \multirow{9}{*}{TBA}    & \multicolumn{1}{c|}{High}                                                       & High                                                        \\ \cline{1-5} \cline{8-9} 
\multicolumn{1}{c|}{\multirow{4}{*}{4 Prepare}}                                                               & Voting                                                                         & \multicolumn{1}{c|}{Medium}                                                     & \multicolumn{1}{c|}{Medium}                                                        & Medium               & \multicolumn{1}{c|}{}                      &                         & \multicolumn{1}{c|}{Low}                                                        & TBA                                                         \\ \cline{2-5} \cline{8-9} 
\multicolumn{1}{c|}{}                                                                                         & \begin{tabular}[c]{@{}c@{}}Inter-communication\\  \& Verification\end{tabular} & \multicolumn{1}{c|}{Low}                                                        & \multicolumn{1}{c|}{Low}                                                           & Low                  & \multicolumn{1}{c|}{}                      &                         & \multicolumn{1}{c|}{Low}                                                        & Low                                                         \\ \cline{2-5} \cline{8-9} 
\multicolumn{1}{c|}{}                                                                                         & No action                                                                      & \multicolumn{1}{c|}{\multirow{2}{*}{TBA}}                                       & \multicolumn{1}{c|}{\multirow{2}{*}{TBA}}                                          & TBA                  & \multicolumn{1}{c|}{}                      &                         & \multicolumn{1}{c|}{\multirow{12}{*}{TBA}}                                      & TBA                                                         \\ \cline{2-2} \cline{5-5} \cline{9-9} 
\multicolumn{1}{c|}{}                                                                                         & Encryption                                                                     & \multicolumn{1}{c|}{}                                                           & \multicolumn{1}{c|}{}                                                              & Low                  & \multicolumn{1}{c|}{}                      &                         & \multicolumn{1}{c|}{}                                                           & High                                                        \\ \cline{1-6} \cline{9-9} 
\multicolumn{1}{c|}{\multirow{4}{*}{5 Commit}}                                                                & Data Backup                                                                    & \multicolumn{1}{c|}{Medium}                                                     & \multicolumn{1}{c|}{Medium}                                                        & Medium               & \multicolumn{1}{c|}{High}                  &                         & \multicolumn{1}{c|}{}                                                           & \multirow{4}{*}{TBA}                                        \\ \cline{2-6}
\multicolumn{1}{c|}{}                                                                                         & \begin{tabular}[c]{@{}c@{}}Inter-communication\\  \& Verification\end{tabular} & \multicolumn{1}{c|}{Low}                                                        & \multicolumn{1}{c|}{Low}                                                           & Low                  & \multicolumn{1}{c|}{High}                  &                         & \multicolumn{1}{c|}{}                                                           &                                                             \\ \cline{2-6}
\multicolumn{1}{c|}{}                                                                                         & No action                                                                      & \multicolumn{1}{c|}{\multirow{2}{*}{TBA}}                                       & \multicolumn{1}{c|}{\multirow{4}{*}{TBA}}                                          & TBA                  & \multicolumn{1}{c|}{TBA}                   &                         & \multicolumn{1}{c|}{}                                                           &                                                             \\ \cline{2-2} \cline{5-6}
\multicolumn{1}{c|}{}                                                                                         & Encryption                                                                     & \multicolumn{1}{c|}{}                                                           & \multicolumn{1}{c|}{}                                                              & Low                  & \multicolumn{1}{c|}{TBA}                   &                         & \multicolumn{1}{c|}{}                                                           &                                                             \\ \cline{1-3} \cline{5-7} \cline{9-9} 
\multicolumn{1}{c|}{\multirow{4}{*}{6 Decide}}                                                                & \textgreater{}1/x Vote                                                         & \multicolumn{1}{c|}{Medium}                                                     & \multicolumn{1}{c|}{}                                                              & \multirow{4}{*}{TBA} & \multicolumn{1}{c|}{High}                  & High                    & \multicolumn{1}{c|}{}                                                           & Low                                                         \\ \cline{2-3} \cline{6-7} \cline{9-9} 
\multicolumn{1}{c|}{}                                                                                         & \textgreater{}1/x Backup                                                       & \multicolumn{1}{c|}{Meidum}                                                     & \multicolumn{1}{c|}{}                                                              &                      & \multicolumn{1}{c|}{High}                  & Medium                  & \multicolumn{1}{c|}{}                                                           & Medium                                                      \\ \cline{2-4} \cline{6-7} \cline{9-9} 
\multicolumn{1}{c|}{}                                                                                         & Self-decision                                                                  & \multicolumn{1}{c|}{High}                                                       & \multicolumn{1}{c|}{High}                                                          &                      & \multicolumn{1}{c|}{Low}                   & Low                     & \multicolumn{1}{c|}{}                                                           & High                                                        \\ \cline{2-4} \cline{6-7} \cline{9-9} 
\multicolumn{1}{c|}{}                                                                                         & No leader                                                                      & \multicolumn{1}{c|}{High}                                                       & \multicolumn{1}{c|}{Low}                                                           &                      & \multicolumn{1}{c|}{High}                  & High                    & \multicolumn{1}{c|}{}                                                           & Medium                                                      \\ \cline{1-7} \cline{9-9} 
\multicolumn{1}{c|}{\multirow{2}{*}{Validators}}                                                              & Fixed                                                                          & \multicolumn{1}{c|}{\multirow{2}{*}{TBA}}                                       & \multicolumn{1}{c|}{Low}                                                           & Low                  & \multicolumn{1}{c|}{High}                  & High                    & \multicolumn{1}{c|}{}                                                           & \multirow{2}{*}{TBA}                                        \\ \cline{2-2} \cline{4-7}
\multicolumn{1}{c|}{}                                                                                         & Dynamic                                                                        & \multicolumn{1}{c|}{}                                                           & \multicolumn{1}{c|}{High}                                                          & High                 & \multicolumn{1}{c|}{Medium}                & Medium                  & \multicolumn{1}{c|}{}                                                           &                                                             \\ \hline
\end{tabular}

\end{table*}

Consensus algorithms can satisfy different CBDC technical features at different levels with a trade-off \cite{groupCBDC}. As a result, they have direct but complex impacts on the CBDC technical features in Section II. 

We reviewed many consensus algorithms\cite{nakamoto2012bitcoin, RAFT, Brown2016,castro1999practical, ethereum,ibft} and found that only part of the difference leads to different consensus algorithms and applications. For example, IBFT\cite{ibft} adopts dynamic set of validators compared with fixed set of validators in PBFT\cite{castro1999practical}. We base on the differences to split the consensus process into different modules. Then we can better analyze the individual impacts on CBDC technical features. Then we combined different components into one algorithm, resulting in its overall impact.

Figure \ref{fig:ConsensusModularization} introduces consensus algorithm components. The following steps describe details about the process:

\begin{enumerate}
    \item Network - Election: a network requires one representative to lead the consensus process before a client sends a transaction request.
    \begin{enumerate}
    \item Voting: the system votes for the leaders. Extensive voting mechanisms involve defining the percentage of votes to become a leader and other requirements. RAFT \cite{RAFT} adopts an election timeout in the voting mechanism to determine the network's leader.
    \item Predetermination: the system predetermines the leaders. For example, the notary node in Corda platform \cite{Brown2016} is determined before network deployment. 
    \item Round-robin: A group of nodes take turns as leaders in a certain order.
    PBFT \cite{castro1999practical} uses round-robin approach to choose the primary (leader).
    \item Proposer: the system has no leader. For example, in some public blockchain systems, a proposer collects transactions from users and proposes them to the network via Proof of Work \cite{nakamoto2012bitcoin}, Proof of Stake \cite{ethereum}. 
    \end{enumerate}
    \item Client - Request: a client submits its payment request to the network. 
    \begin{enumerate}
        \item To one: a client sends the request to only one node. A node is a connection point in a communication network. For example, in RAFT, the client sends the request to another node if they receive no response. Furthermore, the "To one" option can leverage sharding \cite{shard1, shard2} to improve system performance. Sharding means multiple nodes process transactions in parallel and interacts with each other in a specific manner, which we discuss in the example of Section IV.
        \item To all: a client sends the request to all nodes to ensure one node accepts the request in time. 
    \end{enumerate}
    \item Leader / Proposer - Pre-prepare: the leader node or proposer processes the request locally after receiving it.
    \begin{enumerate}
        \item Proof of X: the leader node or proposer leverages one of Proof of X, such as Proof of Work \cite{nakamoto2012bitcoin}, to publish transactions.
        \item Verification: the leader verifies proposed transactions in a specific manner, like checking the input equals the output.
        \item Inter-communication \& Verification: the leader communicates with other nodes and filters transactions. Filtration can reduce the illegal transactions.
    \end{enumerate}
    \item Validators - Prepare: validators vote and communicate with others to verify the request from the leader node.
    \begin{enumerate}
    \item Voting: validators vote for the request.
    \item Inter-communication \& Voting: validators communicate with each other and vote for the request. PBFT \cite{castro1999practical} adopts this method to prevent malicious behaviors.
     \item No action: validators do nothing. For example, validators (called followers) in RAFT \cite{RAFT} do nothing.
    \end{enumerate}
    \item Validators - Commit: validators record transactions in their database. They could directly undertake data backup, like RAFT \cite{RAFT}, or communicate with each other and then determine whether to record the proposed request, like PBFT \cite{castro1999practical}.
    \begin{enumerate}
       \item Data backup: validators make a backup in their local databases.
       \item Inter-communication \& Data backup: validators communicate before backup.
    \end{enumerate}
    \item Leader / Proposer - Decide: the leader finalizes the request.  
    \begin{enumerate}
        \item x\% backup replies: the leader node receives more than x\% backup replies. For example, in RAFT\cite{RAFT}, the leader node needs to receive 50\% replies before responding to its client.
        \item x\% vote replies: the leader node receives more than x\% vote replies then finalizes the transaction.
        \item Self-decision: the leader node decides the transaction by itself. For example, the notary node in Corda \cite{Brown2016} verifies proposed transactions by itself.
    \end{enumerate}
    \end{enumerate}
    
Besides consensus process options, we include another two options to improve the overall algorithm:
    \begin{enumerate}
        \item Fixed / Dynamic Validators: validators are non-leader nodes that can participate in consensus. Some consensus algorithms require all nodes to participate in consensus, while others need selected or random dynamic nodes. For example, RAFT \cite{RAFT} and PBFT \cite{castro1999practical} need all nodes to participate data backup, while IBFT \cite{ibft} adopts a dynamic set of validators. A dynamic set of validators can provide a higher performance because of more flexible participants. In contrast,  fixed validators are easier to implement and more secure.
    \item Encryption: CBDC designers can use encryption in every step to secure transmitted information but decrease performance. It is independent of previous options. In most algorithms, consensus algorithms are more associated with achieving consistency and ignore encryption. However, in the CBDC scenario, encryption plays a role, so we include it in the consensus process.
\end{enumerate}


Each component has many extensions, and components have constraints between each other. For example, "proof of X" in the third step is possibly connected with "proposer" in the first step. 

Table \ref{Table:Impact} shows how the above components influence the mentioned CBDC technical features in Section II. We measure the impact through categories of High\footnote{High means the module positively impacts the dimension}, Medium\footnote{Medium means the module has no impact or relative medium impact} and Low\footnote{Low means the module has a relatively harmful impact on the dimension}. TBA indicates to be further analyzed. Every component has its impact on every CBDC dimension. We measure these impacts by empirical experiments and formal proofs. Different components together can build one consensus algorithm with the same weight. If needed, they could have different weights of impacts on different CBDC technical features. 

\begin{figure*}[h]
	\centering
	\includegraphics[width=1.0\textwidth]{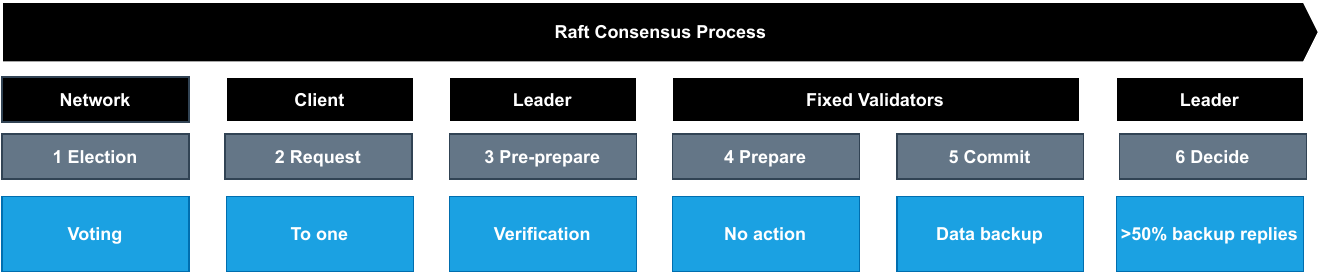} 
	\caption{\textbf{Consensus Process of RAFT starts from voting and election timeout mechanism in the leader election. Then a client sends a request to the leader in the network. After the leader's verification, validators make a backup and respond to the leader node. Then this transaction will be regarded as valid when receiving more than half of the notices.}  \label{fig:RAFT}}
\end{figure*}

\begin{table*}[h]
\centering
\caption{\textbf{Impact of RAFT on CBDC technical features. Total adds all values in the table with High (+1), Medium or TBA (0), Low (-1). }}
\label{Table:RAFT}
\begin{tabular}{cc|ccc|cc|cc}
\hline
\multicolumn{2}{c|}{\multirow{2}{*}{\textbf{Modules}}}             & \multicolumn{3}{c|}{\textbf{Performance}}                                                                             & \multicolumn{2}{c|}{\textbf{Security}}                             & \multicolumn{2}{c}{\textbf{Privacy}}                                                                                                         \\ \cline{3-9} 
\multicolumn{2}{c|}{}                                              & \multicolumn{1}{c|}{\textbf{User Scalability}} & \multicolumn{1}{c|}{\textbf{Network Scalability}} & \textbf{Latency} & \multicolumn{1}{c|}{\textbf{Resilience}} & \textbf{Cyber-Security} & \multicolumn{1}{c|}{\begin{tabular}[c]{@{}c@{}}\textbf{Customer}\\ \textbf{Privacy}\end{tabular}} & \begin{tabular}[c]{@{}c@{}}\textbf{Business}\\ \textbf{Secrecy}\end{tabular} \\ \hline
\multicolumn{1}{c|}{1 Leader Election} & Voting                    & \multicolumn{1}{c|}{TBA}                       & \multicolumn{1}{c|}{Medium}                       & Medium           & \multicolumn{1}{c|}{TBA}                 & High                    & \multicolumn{1}{c|}{TBA}                                                        & TBA                                                        \\ \hline
\multicolumn{1}{c|}{2 Request}         & To one                    & \multicolumn{1}{c|}{High}                      & \multicolumn{1}{c|}{High}                         & Low              & \multicolumn{1}{c|}{TBA}                 & TBA                     & \multicolumn{1}{c|}{Medium}                                                     & TBA                                                        \\ \hline
\multicolumn{1}{c|}{3 Pre-Prepare}     & Verification              & \multicolumn{1}{c|}{High}                      & \multicolumn{1}{c|}{High}                         & High             & \multicolumn{1}{c|}{TBA}                 & TBA                     & \multicolumn{1}{c|}{TBA}                                                        & TBA                                                        \\ \hline
\multicolumn{1}{c|}{4 Prepare}         & No action                 & \multicolumn{1}{c|}{TBA}                       & \multicolumn{1}{c|}{TBA}                          & TBA              & \multicolumn{1}{c|}{TBA}                 & TBA                     & \multicolumn{1}{c|}{TBA}                                                        & TBA                                                        \\ \hline
\multicolumn{1}{c|}{5 Commit}          & Data Backup               & \multicolumn{1}{c|}{Medium}                    & \multicolumn{1}{c|}{Medium}                       & Medium           & \multicolumn{1}{c|}{High}                & TBA                     & \multicolumn{1}{c|}{TBA}                                                        & TBA                                                        \\ \hline
\multicolumn{1}{c|}{6 Decide}          & \textgreater{}50\% Backup & \multicolumn{1}{c|}{Meidum}                    & \multicolumn{1}{c|}{TBA}                          & TBA              & \multicolumn{1}{c|}{High}                & Medium                  & \multicolumn{1}{c|}{TBA}                                                        & Medium                                                     \\ \hline
\multicolumn{1}{c|}{Validators}        & Fixed                     & \multicolumn{1}{c|}{TBA}                       & \multicolumn{1}{c|}{Low}                          & Low              & \multicolumn{1}{c|}{High}                & High                    & \multicolumn{1}{c|}{TBA}                                                        & TBA                                                        \\ \hline
\multicolumn{2}{c|}{\textbf{Total}}                                & \multicolumn{1}{c|}{2}                         & \multicolumn{1}{c|}{1}                            & 1                & \multicolumn{1}{c|}{3}                   & 2                       & \multicolumn{1}{c|}{0}                                                          & 0                                                          \\ \hline
\multicolumn{2}{c|}{\textbf{Average}}                              & \multicolumn{3}{c|}{\textbf{4/3}}                                                                                     & \multicolumn{2}{c|}{\textbf{2.5}}                                  & \multicolumn{2}{c}{\textbf{0}}                                                                                                               \\ \hline
\end{tabular}
\end{table*}

We have referenced RAFT several times. Here we use the Consensus Process map to describe the RAFT consensus algorithm as seen in figure \ref{fig:RAFT}. Then we leverage Table \ref{Table:RAFT} to show that the RAFT consensus algorithm can provide CBDC systems good fault-tolerance and performance while it takes no privacy protection measurements.

Overall, the evaluation sub-framework can guide CBDC designers to consider related factors, analyze different combinations, and find a suitable consensus algorithm. However, we need a method to judge whether the combination meets the expectations. Therefore, we need the verification sub-framework to ensure proposed consensus algorithms are valid and satisfactory.

\subsubsection{Operating Architecture}
We propose two new operating architectures based on current CBDC architectures (figure 1). The operating architecture determines how the network functions at a high level. As mentioned before, a trade-off\cite{groupCBDC} between CBDC technical features exists that we can not leverage consensus algorithms to provide excellent performance, security, and privacy at the same time. However, we can update the operating architecture to compensate for weak CBDC technical features. In this part, we use business secrecy as an example.

We concluded a three-tier CBDC architecture (Figure \ref{fig:three-tier}) to describe institutions that do not become tier-1 but want to provide CBDC services. The model describes how tier-1.5 institutions provide CBDC services to their customers in the current situation.

Tier-1.5 institutions have to provide transaction information to tier-1 ones for bookkeeping because tier-1 institutions operate the ledgers. Once tier-1 institution records the transaction in the ledger, the transaction becomes valid. However, tier-1.5 institutions will refuse to provide the customer data to tier-1 ones because they are competitors with interest conflict. Some commercial institutions even give up providing CBDC-related services. 

\begin{figure}[h]
	\centering
	\includegraphics[width=0.500\textwidth]{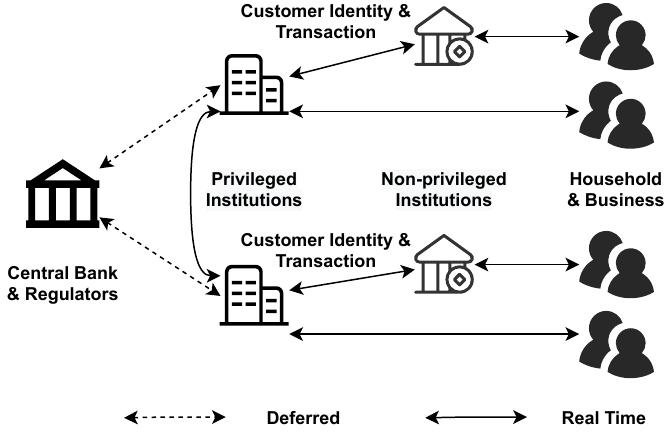} 
	\caption{\textbf{Three-tier CBDC architecture}  \label{fig:three-tier}}
\end{figure}

To safeguard tier-1.5 institutions from data monopoly, we propose two operating architectures:
\begin{enumerate}
    \item Use dynamic virtual addresses to keep the identities of participants secret from tier-1 institutions;
    \item Use an independently operating organization that has no conflict of interest;
\end{enumerate}

\begin{figure}[h]
	\centering
	\includegraphics[width=0.500\textwidth]{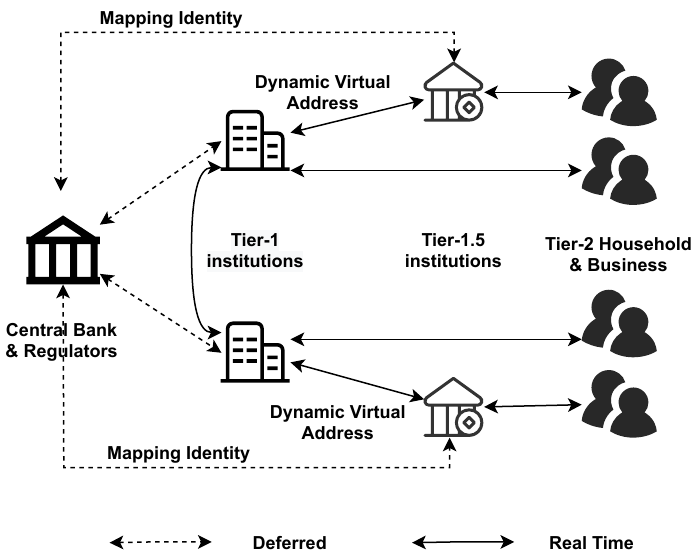} 
	\caption{\textbf{Operating Architecture one leverages dynamic virtual address to avoid tier-1 institutions from knowing the real identities of customers.}  \label{fig:Operating Model Option one}}
\end{figure}

Figure \ref{fig:Operating Model Option one} shows tier-1.5 institutions create virtual addresses for their customers in the tier-1 ledger. Privacy includes identity and transaction information. Virtual addresses ensure that tier-1 institutions have no access to the identity information of the payee and payer. Tier-1.5 institutions provide regulators with a mapping table between virtual addresses and real identities. Only regulators and tier-1.5 institutions can know the mapping relation of virtual addresses to real identities. Tier-1.5 institutions only need to inform the central bank of the identity information. Then the central bank could get transaction information from tier-1 institutions by combining identity mapping and ledger transaction information.

However, tier-1 institutions can infer identity information by analyzing enough token flows and real-world events even though they only know transaction information. To further protect the business secrecy of tier-1.5 institutions, we can make virtual addresses dynamic that tier-1.5 institutions create new virtual addresses to collect changes in transactions. This technical solution can prevent tier-1 institutions from accessing tier-1.5 institutions' customer data further.

There are other methods to improve the operating architecture. For example, figure \ref{fig:Operating Model Option Two} presents a third-party organization in tier-1 to ensure no competition with tier-1s and tier-2s. However, the operating organization needs a feasible business model to ensure enough money to support its stable operation. 
\begin{figure}[h]
	\centering
	\includegraphics[width=0.500\textwidth]{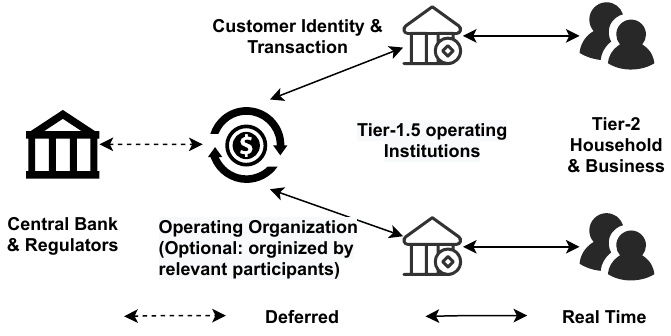}
	\caption{\textbf{Operating Architecture two builds an operating organization rather than tier-1 institutions to operate the CBDC system without interest conflict. The operating organization can know the transaction data of tier-1.5 institutions without business secrecy issues.} \label{fig:Operating Model Option Two}}
\end{figure}

\subsection{Verification Sub-Framework}
By the evaluation sub-framework, CBDC designers can develop a customized solution. Then the verification sub-framework works to ensure the proposed solution's feasibility and rationality. 

The verification sub-framework contains diverse methods to verify proposed solutions. We divide these methods into two categories: empirical experiments and formal proofs. For empirical experiments, we can simulate a real scenario and test parameters in the full load environment. For formal proofs, we can build a mathematical model and infer related theories and find whether they meet initial expectations. 

The verification sub-framework works in the following procedures:
\begin{enumerate}
    \item Model proposed solutions for further verification. The verification sub-framework can guide CBDC designers to describe the proposed solution mathematically. 
    \item Follow the built model to conduct empirical experiments to examine performance.
    \item Follow the built model to conduct formal proofs, mainly for security and privacy.
\end{enumerate} 

\section{CBDC Scenario Example}
Next, we present an example of using the CEV framework. We assume a country with a large population and well-developed technology and communication. The CBDC designers focus on privacy and performance, especially network scalability, latency, business secrecy. Then we leverage the CEV framework to propose a suitable solution for this virtual country.

\subsection{Evaluation}
We leverage the evaluation sub-framework to choose the operating architecture and propose several consensus algorithms. In this example, both operating architectures are feasible to improve business secrecy. So we choose figure \ref{fig:Operating Model Option one} as the operating architecture. Then the evaluation sub-framework can propose consensus algorithms for different consensus networks. 

Since the example emphasizes performance among the three CBDC technical features. We leverage table 1 to find the combinations with relatively high scores in performance for both wholesale and retail consensus networks. 
\begin{figure*}[h]
	\centering
	\includegraphics[width=1.0\textwidth]{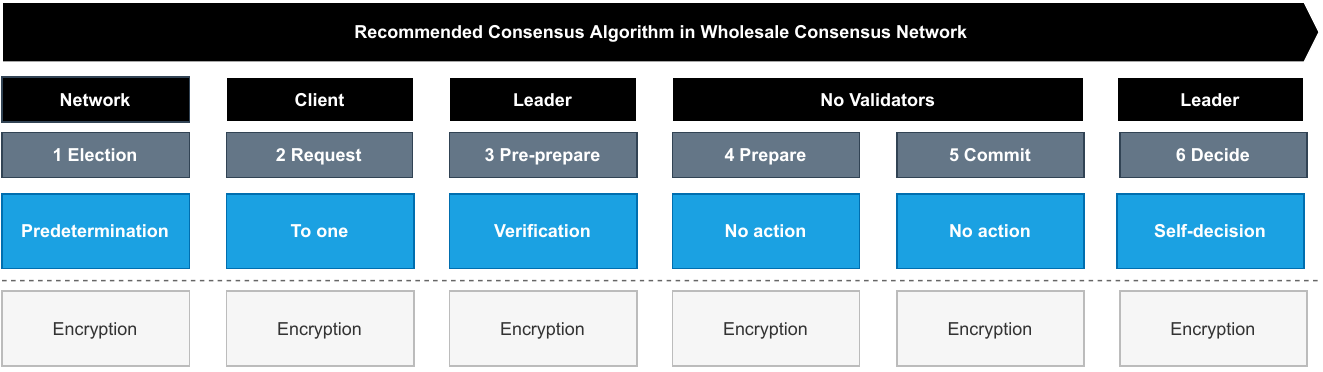}
	\caption{\textbf{Recommended Consensus Algorithm in the Wholesale Consensus Network with following steps: 
	1. The network predetermines the leader nodes that the central bank operates. 2. The client sends a cross-shard transaction request to its sharded leader (tier-1 institution), and then the tier-1 institution forward it to the leader node (central bank) in the wholesale network. 3. The leader node (central bank) verifies the transaction and finishes the related issuance and redemption transactions in different retail networks. Afterwards, the transaction is successful.} \label{fig:Recommended Consensus Algorithm One}}
\end{figure*}

The recommended consensus algorithm in the wholesale consensus network works, and retail consensus networks work in figure \ref{fig:Recommended Consensus Algorithm One} and figure \ref{fig:Recommended Consensus Algorithm Two}, respectively. Table 3 and Table 4 are derived from table 1. Table 3 shows that the recommended consensus algorithm in the wholesale consensus network has a good performance, especially network scalability, but the algorithm may bring a potential security issue. 
Similarly, table 4 shows that the consensus algorithm in the retail consensus network has a good performance and a potential security issue. Both tables present that the system has relative good privacy. The impact table only provides a rough analysis to CBDC designers, so we need to further verify the proposed solution in the verification sub-framework and make sure proposed solutions are objective and reasonable. 

\begin{figure*}[h]
	\centering
	\includegraphics[width=1.0\textwidth]{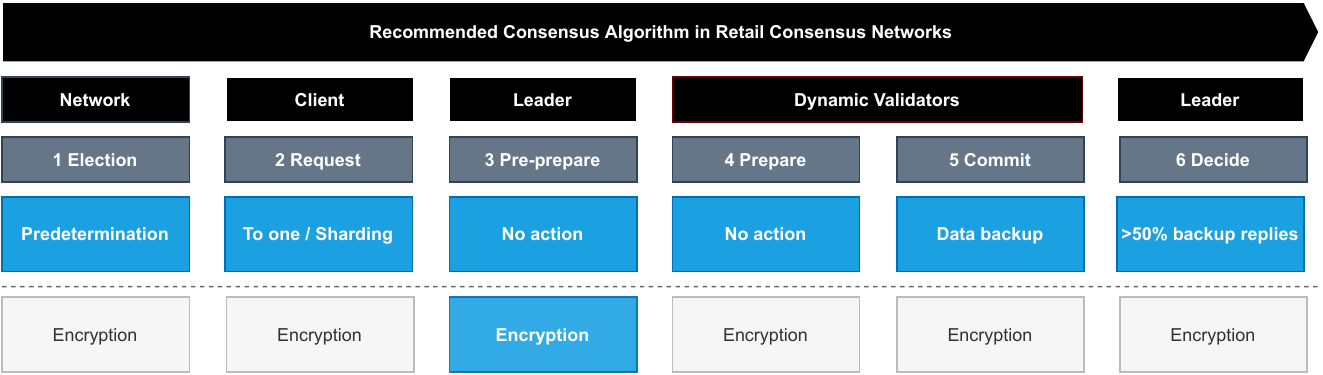}
	\caption{\textbf{Recommended Consensus Algorithm in the Retail Consensus Networks with the following steps:     
	1. The network predetermines the leader nodes that tier-1 institutions run.
    2. The client sends transactions to its sharded leader.
    3. The leader node encrypts transactions before sending them to validators in the network (possible competitors). 
    4. The network leverages a dynamic set of validators to encrypt backup data. 
    5. The leader node sends a success notice to the client after receiving 50 per cent of notices from validators.} \label{fig:Recommended Consensus Algorithm Two}}
\end{figure*}
\begin{table*}[h]
\centering
\caption{\textbf{Impact of the Recommended Consensus Algorithm in the Wholesale Consensus Network. Total adds all values in the table with High (+1), Medium or TBA (0), Low (-1).}}
\label{Table:wholesale}
\begin{tabular}{cc|ccc|cc|cc}
\hline
\multicolumn{2}{c|}{\multirow{2}{*}{\textbf{Modules}}}    & \multicolumn{3}{c|}{\textbf{Performance}}                                                                                                              & \multicolumn{2}{c|}{\textbf{Security}}                             & \multicolumn{2}{c}{\textbf{Privacy}}                                                                                                          \\ \cline{3-9} 
\multicolumn{2}{c|}{}                                     & \multicolumn{1}{c|}{\textbf{User Scalability}} & \multicolumn{1}{c|}{\begin{tabular}[c]{@{}c@{}}\textbf{Network}\\ \textbf{Scalability}\end{tabular}} & \textbf{Latency} & \multicolumn{1}{c|}{\textbf{Resilience}} & \textbf{Cyber-Security} & \multicolumn{1}{c|}{\begin{tabular}[c]{@{}c@{}}\textbf{Customer}\\ \textbf{Privacy}\end{tabular}} & \begin{tabular}[c]{@{}c@{}}\textbf{Business}\\  \textbf{Secrecy}\end{tabular} \\ \hline
\multicolumn{1}{c|}{1 Leader Election} & Predetermination & \multicolumn{1}{c|}{TBA}                       & \multicolumn{1}{c|}{High}                                                          & High             & \multicolumn{1}{c|}{TBA}                 & TBA                     & \multicolumn{1}{c|}{TBA}                                                        & High                                                        \\ \hline
\multicolumn{1}{c|}{2 Request}         & To one           & \multicolumn{1}{c|}{High}                      & \multicolumn{1}{c|}{High}                                                          & Low              & \multicolumn{1}{c|}{TBA}                 & TBA                     & \multicolumn{1}{c|}{Medium}                                                     & TBA                                                         \\ \hline
\multicolumn{1}{c|}{3 Pre-Prepare}     & Verification     & \multicolumn{1}{c|}{High}                      & \multicolumn{1}{c|}{High}                                                          & High             & \multicolumn{1}{c|}{TBA}                 & TBA                     & \multicolumn{1}{c|}{TBA}                                                        & TBA                                                         \\ \hline
\multicolumn{1}{c|}{4 Prepare}         & No action        & \multicolumn{1}{c|}{TBA}                       & \multicolumn{1}{c|}{TBA}                                                           & TBA              & \multicolumn{1}{c|}{TBA}                 & TBA                     & \multicolumn{1}{c|}{TBA}                                                        & TBA                                                         \\ \hline
\multicolumn{1}{c|}{5 Commit}          & No action        & \multicolumn{1}{c|}{TBA}                       & \multicolumn{1}{c|}{TBA}                                                           & TBA              & \multicolumn{1}{c|}{TBA}                 & TBA                     & \multicolumn{1}{c|}{TBA}                                                        & TBA                                                         \\ \hline
\multicolumn{1}{c|}{6 Decide}          & Self-decision    & \multicolumn{1}{c|}{High}                      & \multicolumn{1}{c|}{High}                                                          & TBA              & \multicolumn{1}{c|}{Low}                 & Low                     & \multicolumn{1}{c|}{TBA}                                                        & High                                                        \\ \hline
\multicolumn{1}{c|}{Validators}        & Fixed            & \multicolumn{1}{c|}{TBA}                       & \multicolumn{1}{c|}{Low}                                                           & Low              & \multicolumn{1}{c|}{High}                & High                    & \multicolumn{1}{c|}{TBA}                                                        & TBA                                                         \\ \hline
\multicolumn{2}{c|}{\textbf{Total}}                       & \multicolumn{1}{c|}{3}                         & \multicolumn{1}{c|}{3}                                                             & 0                & \multicolumn{1}{c|}{0}                   & 0                       & \multicolumn{1}{c|}{0}                                                          & 2                                                           \\ \hline
\multicolumn{2}{c|}{\textbf{Average}}                     & \multicolumn{3}{c|}{\textbf{2}}                                                                                                                        & \multicolumn{2}{c|}{\textbf{0}}                                    & \multicolumn{2}{c}{\textbf{1}}                                                                                                                \\ \hline
\end{tabular}
\end{table*}
\begin{table*}[h]
\centering
\caption{\textbf{Impact of the Recommended Consensus Algorithm in the Retail Consensus Network. Total adds all values in the table with High (+1), Medium or TBA (0), Low (-1).}}
\label{Table:Retail}
\begin{tabular}{cc|ccc|cc|cc}
\hline
\multicolumn{2}{c|}{\multirow{2}{*}{\textbf{Modules}}}                 & \multicolumn{3}{c|}{\textbf{Performance}}                                                                                                                                               & \multicolumn{2}{c|}{\textbf{Security}}                             & \multicolumn{2}{c}{\textbf{Privacy}}                                                                                                          \\ \cline{3-9} 
\multicolumn{2}{c|}{}                                                  & \multicolumn{1}{c|}{\begin{tabular}[c]{@{}c@{}}\textbf{User}\\ \textbf{Scalability}\end{tabular}} & \multicolumn{1}{c|}{\begin{tabular}[c]{@{}c@{}}\textbf{Network}\\ \textbf{Scalability}\end{tabular}} & \textbf{Latency} & \multicolumn{1}{c|}{\textbf{Resilience}} & \textbf{Cyber-Security} & \multicolumn{1}{c|}{\begin{tabular}[c]{@{}c@{}}\textbf{Customer}\\ \textbf{Privacy}\end{tabular}} & \begin{tabular}[c]{@{}c@{}}\textbf{Business}\\  \textbf{Secrecy}\end{tabular} \\ \hline
\multicolumn{1}{c|}{1 Election}                     & Predetermination & \multicolumn{1}{c|}{TBA}                                                        & \multicolumn{1}{c|}{High}                                                          & High             & \multicolumn{1}{c|}{TBA}                 & TBA                     & \multicolumn{1}{c|}{TBA}                                                        & High                                                        \\ \hline
\multicolumn{1}{c|}{2 Request}                      & To one           & \multicolumn{1}{c|}{High}                                                       & \multicolumn{1}{c|}{High}                                                          & Low              & \multicolumn{1}{c|}{TBA}                 & TBA                     & \multicolumn{1}{c|}{Medium}                                                     & TBA                                                         \\ \hline
\multicolumn{1}{c|}{\multirow{2}{*}{3 Pre-prepare}} & No action        & \multicolumn{1}{c|}{High}                                                       & \multicolumn{1}{c|}{High}                                                          & High             & \multicolumn{1}{c|}{Medium}              & Low                     & \multicolumn{1}{c|}{Medium}                                                     & Medium                                                      \\ \cline{2-9} 
\multicolumn{1}{c|}{}                               & Encryption       & \multicolumn{1}{c|}{TBA}                                                        & \multicolumn{1}{c|}{TBA}                                                           & Low              & \multicolumn{1}{c|}{TBA}                 & TBA                     & \multicolumn{1}{c|}{High}                                                       & High                                                        \\ \hline
\multicolumn{1}{c|}{4 Prepare}                      & No action        & \multicolumn{1}{c|}{TBA}                                                        & \multicolumn{1}{c|}{TBA}                                                           & TBA              & \multicolumn{1}{c|}{TBA}                 & TBA                     & \multicolumn{1}{c|}{TBA}                                                        & TBA                                                         \\ \hline
\multicolumn{1}{c|}{5 Commit}                       & Data Backup      & \multicolumn{1}{c|}{Medium}                                                     & \multicolumn{1}{c|}{Medium}                                                        & Medium           & \multicolumn{1}{c|}{High}                & TBA                     & \multicolumn{1}{c|}{TBA}                                                        & TBA                                                         \\ \hline
\multicolumn{1}{c|}{6 Decide}                       & Self-decision    & \multicolumn{1}{c|}{High}                                                       & \multicolumn{1}{c|}{High}                                                          & TBA              & \multicolumn{1}{c|}{Low}                 & Low                     & \multicolumn{1}{c|}{TBA}                                                        & High                                                        \\ \hline
\multicolumn{1}{c|}{Validators}                     & Dynamic          & \multicolumn{1}{c|}{TBA}                                                        & \multicolumn{1}{c|}{High}                                                          & High             & \multicolumn{1}{c|}{Medium}              & Medium                  & \multicolumn{1}{c|}{TBA}                                                        & TBA                                                         \\ \hline
\multicolumn{2}{c|}{\textbf{Total}}                                    & \multicolumn{1}{c|}{3}                                                          & \multicolumn{1}{c|}{5}                                                             & 1                & \multicolumn{1}{c|}{0}                   & -2                      & \multicolumn{1}{c|}{1}                                                          & 3                                                           \\ \hline
\multicolumn{2}{c|}{\textbf{Average}}                                  & \multicolumn{3}{c|}{\textbf{3}}                                                                                                                                                         & \multicolumn{2}{c|}{\textbf{-1}}                                   & \multicolumn{2}{c}{\textbf{2}}                                                                                                                \\ \hline
\end{tabular}
\end{table*}

Specifically, encryption in the third step protects business secrecy from validators because validators only backup encrypted data for tampering with proof in the algorithm. Additionally, "To one / Sharding" in the "client-request" step can improve the system's performance in a token-based sharding method. Besides, "Data backup" can increase the resilience of CBDC systems. These impacts on CBDC technical features are derived from the verification sub-framework rather than our subjective idea. Next, we use the verification sub-framework to see whether the proposed solution meets the initial expectations.
\subsection{Verification}
Then we use the verification sub-framework to verify the above-proposed solution by building a theoretical model for the solution and following the model to carry out empirical experiments and formal proofs. 

\begin{figure}[h]
	\centering
	\includegraphics[width=0.500\textwidth]{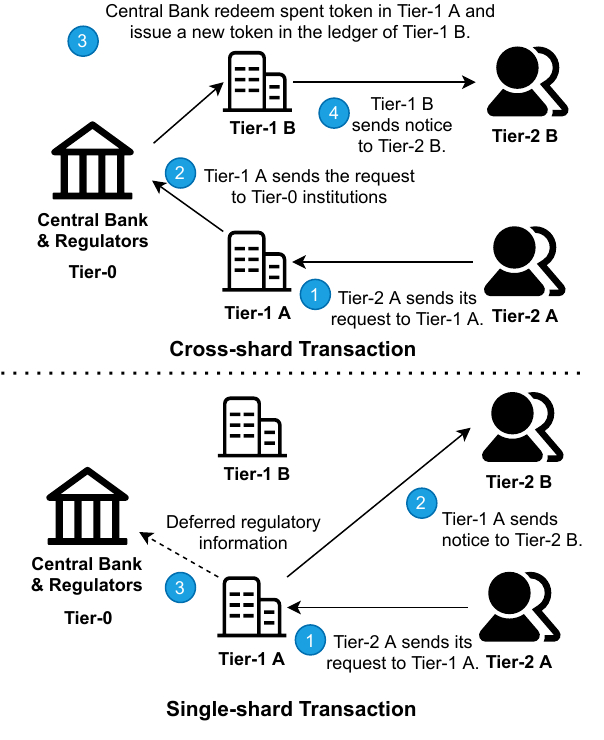}
	\caption{\textbf{Cross-shard transaction describes a transaction with two ledgers involved; single-shard transaction describes a transaction within one ledger.} \label{fig:token-based}}
\end{figure} 

We first introduce the sharding method in our experiment. The previous work \cite{shard1, shard2,ShardPerformance} discusses account-based sharding and token-based sharding. We here introduce them in a CBDC scenario with two types of transactions involved (shown in figure \ref{fig:token-based}).

Account-based sharding divide users by accounts. In the two-tier CBDC architecture, tier-1 institutions have different wallets and divide users by their accounts. For example, Tier-2 A comes to Tier-1 A\footnote{Tier-1 A is an example tier-1 institution.} when using CBDC because Tier-2 A created its account from Tier-1 A. A cross-shard transaction happens if Tier-2 A transfers its money to Tier-1 B's customers. The cross-shard transaction needs the currency issuer (the central bank) to redeem a token in one ledger and issue a new token in another ledger, which decreases the system's performance a lot compared with single-shard transactions\cite{ShardPerformance}. We use token-based sharding as an example.

Token-based sharding divides users by tokens. If Tier-1 A issues a token to Tier-2 A, Tier-2 A comes to Tier-1 A to initiate CBDC related services because of the token (figure \ref{fig:tokendata}) that Tier-1 A operates. Users only need to contact their token service provider and transfer the token to everyone without a cross-shard transaction in a token-based shard mode. However, token-based sharding needs the central bank and operating institutions to provide a uniform interface to distribute transactions. Because if a customer who use Tier-1 A's token come to Tier-1 B's interface, the uniform interface should distribute the transaction to Tier-1 A.

\begin{figure}[h]
	\centering
	\includegraphics[width=0.500\textwidth]{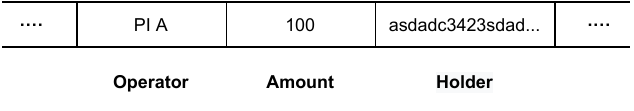}
	\caption{\textbf{Holder is the identifier to determine which tier-1 institutions that tier-2 customers belong to.} \label{fig:tokendata}}
\end{figure} 

\subsubsection{Model}
Next, we built a CBDC state machine for the recommended algorithms to see how the proposed solution functions.

Figure \ref{fig:mathmodel} shows ledger state machine. The model provides a formal description of finite state machine M = (S, V, t). The state machine can describe any state $s \in S$ for every moment. It can read input token $\tau \in V$ and proceed to the next state by different transitions t(s, $\tau$). 

\begin{figure}[h]
	\centering
	\includegraphics[width=0.50\textwidth]{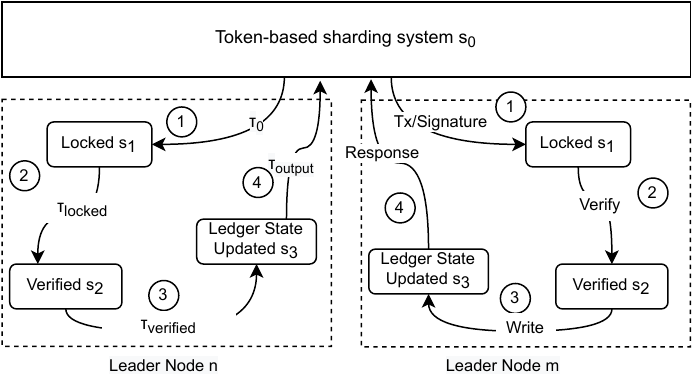}
	\caption{\textbf{State Machine describes how a transaction is being executed, including token's state on the left side and operations on the right side: In the initial machine state $s_{0}$, a token-based sharding CBDC system distributes transactions to different leaders (tier-1 institutions). The leader checks transaction signatures and input tokens $\tau_{0}$. After the shard leader verifies the signature, the tokens become locked. If no leader has previously locked the tokens, the output becomes $\tau_{locked}$ and the machine moves to the locked state $s_{1}$. Otherwise, it would exit (a rolled-back transaction to the initial state). Next, the leader verifies whether $\tau_{locked}$ is available. If so, the output token becomes $\tau_{verified}$, and the machine moves to the verified state $s_{2}$. Otherwise, it would exit. Finally, the leader writes the transaction with inputs and outputs. If successful, the output would be $\tau_{output}$, and the machine moves into the state $s_{3}$ and soon back to $s_{0}$ again.} \label{fig:mathmodel}}
\end{figure} 

Figure \ref{fig:mathmodel} shows how the state machine ensures the leader node records the token in the ledger before noticing clients. However, it only covers operations from the leader node rather than validator nodes. The validators in the consensus algorithm do a data backup, which can help to reduce malicious behaviours. We discuss it in Section IV 4 Security.
\begin{table}
\caption{\textbf{Table of Notations}}
\centering
\begin{tabular}{@{}ll@{}}
\toprule
\textbf{Notation} & \textbf{Definition}  \\ \midrule
$x, y, z$  &  time  \\
$x R y$ & x before y\\ 
$\tau$  & Token   \\
$V$  & Token Set   \\
$p(\tau)$  & The leader has recorded $ \tau$    \\
$Tx(\tau,x)$  &  Transaction with input $\tau$ at time x  \\
$H(\tau, x)$ & $\tau$ has been spent before $\mathrm{x}$   \\
$F(\tau, x)$  & $ \tau$ can be spent after $x$   \\
$r(\tau)$  & Received token in a transaction using $ \tau$ \\
$c(\tau)$  & Change token in a transaction using $ \tau$ \\
$f(\tau)$  & Received token in a cross-shard transaction using $ \tau$ \\
$d_{G}^{+}(\tau)$ & The out-degree of vertex v\\
$d_{G}^{+}(\tau,x)$ & The out-degree of vertex v at time x\\
H  & Leaders execute transactions by state machine flow\\
EE  & Existential Elimination \\
AE  & And Elimination \\
AI  & And Introduction \\
UE  & Universal Elimination \\
UI  & Universal Introduction \\
IE  & Implication Elimination \\
II  & Implication Introduction \\
BE  & Biconditional Elimination \\
Ind & Induction \\
AS & Assumption \\
\bottomrule
\end{tabular}
\end{table}

Figure \ref{fig:figuretheory} shows data model of leaders' ledgers. Transactions are subject to the leader's ledger. Once a leader updates its ledger successfully, the transaction becomes legal and immutable. 

\begin{definition}
\textbf{(Ledger State)} The ledger state of the CBDC is defined as a directed graph D=<V(D), E(D), $\varphi$> that the elements of V(D) are vertices (tokens) and the elements of E(D) are edges (token flow). $\varphi$ is ordered mapping from token set V to token flow set E.
\end{definition}
\begin{figure}[h] 
	\centering
	\includegraphics[width=0.50\textwidth]{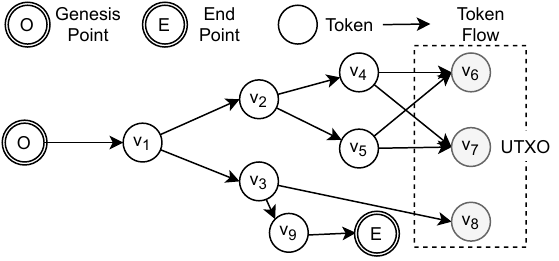}
	\caption{\textbf{Ledger State. v are vertices (tokens), belonging to V, edge represents token flow. On the right side, some tokens without edge coming from form UTXO.} \label{fig:figuretheory}}
\end{figure}

\begin{definition}
\textbf{(UTXO)} $UTXO=\{\tau|\tau\in V(D) \land d_{G}^{+}(\tau)=0\}$.
UTXO is an unspent transaction output set. An unspent token means no edge coming from it (the out-degree of it is 0).
\end{definition}

\begin{definition}
\textbf{(Transaction Graph)} A transaction graph is a directed graph TD=<V(TD), E(TD), $\varphi$>. The in-degree of an input token and the out-degree of an output token are both 0 in any transaction graph. The leader node updates the ledger when the state machine finishes a transaction ($\forall$x. $\forall$ $\tau$.(Tx($\tau$,x) $\Rightarrow$ D = D + TD)).
\end{definition} 

Figure \ref{fig:transactiontype} shows all types of transaction graph. Only currency issuer (Central Bank) can initiate the Initial Issuance Transaction, generating a new token from the genesis point. It also needs to check the Final Redemption Transaction as the transaction receiver. Both types of transactions need the central bank to carry out real-time operations. 

\begin{figure}[h]
	\centering
	\includegraphics[width=0.50\textwidth]{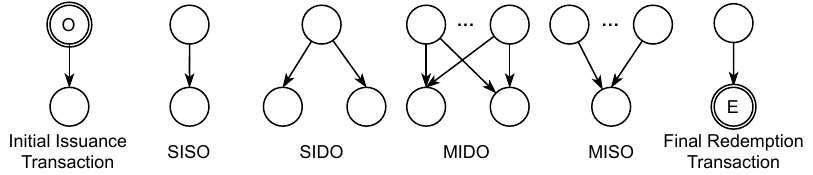}
	\caption{\textbf{Different types of transaction graphs. SI refers to single-input. MI refers to multi-input. SO refers to single-output. DO refers to double-output.} \label{fig:transactiontype}}
\end{figure} 

The following mathematical expressions present the token flows rather than transactions. Here we add an assumption that the leader nodes are non-faulty (H) and follow the model procedures. Faulty nodes may behave arbitrarily and be vulnerable to inside and outside attacks. With non-faulty nodes, we can ensure leader nodes record tokens in the ledger in every transaction:
\begin{enumerate}
    \item  $\forall \tau. \, (H $$\land$$p(\tau) \Rightarrow p(r(\tau)) $$\land$$p(c(\tau)))$
    \item  $\forall \tau. \, (H $$\land$$p(\tau) \Rightarrow p(f(\tau)))$
\end{enumerate}

$\forall \tau. \, (H $$\land$$p(\tau) \Rightarrow p(r(\tau))$ represents that if the ledger has recorded input $\tau$, a non-faulty leader(H) always adds a valid transaction graph to the ledger with inputs of $\tau$ and outputs of received token $r(\tau)$ and change token $c(\tau)$. If the change is 0 ($c(\tau) = null$), p(c($\tau$)) means no token recorded. 

$\forall \tau. \, (H $$\land$$p(\tau) \Rightarrow p(f(\tau)))$ represents central banks and non-faulty leaders(H) ensure that the receiving ledger record the output tokens in a cross-shard transaction (figure \ref{fig:cross-shard transaction}). 

Cross-shard transactions will not happen if we split one into several concurrent sub-transactions. However, cross-shard transactions match some business scenarios and possibly happen. For example, CBDC users may pay tokens in different ledgers in an atomic transaction. Alternatively, CBDC designers want to control the number of tokens and consolidate tokens from different ledgers to one new token. In a CBDC system, the central bank is responsible for issuing and redeeming tokens, regulating both transactions and ensuring that the new shard records output tokens.

\begin{figure}[h]
	\centering
	\includegraphics[width=0.500\textwidth]{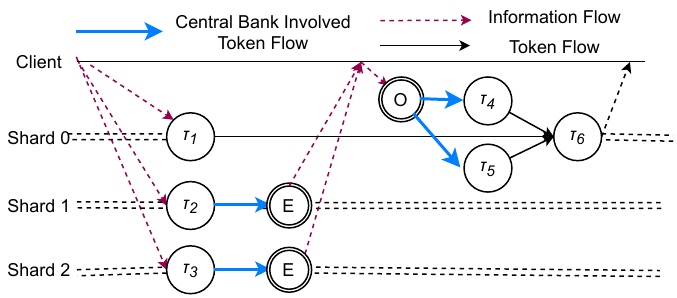}
	\caption{\textbf{Cross-shard Transaction: a client has allocated different tokens in three different shards and used them to initiate a transaction. The transaction first turns the input tokens to the endpoint via the Final Redemption Transaction and issues new tokens in the new shard via the Initial Issuance Transaction.} \label{fig:cross-shard transaction}}
\end{figure} 

\subsubsection{Performance}
Empirical experiments test user scalability, network scalability, and latency. To ensure experiments are close to reality, we randomly initiate user transactions by different payment methods, including face-to-face transfer, collecting, etc.

We leverage AWS EC2 to deploy CBDC networks and carry out the empirical experiments shown in figure \ref{fig:Performance}. Unfortunately, since the Corda open-source version has limitations on transaction volume, we could not demonstrate an extremely large TPS (transaction per second) in the experiment due to cost control. However, the experiment demonstrates performance improvement in a CBDC system.

\begin{figure}[h]
	\centering
	\includegraphics[width=0.50\textwidth]{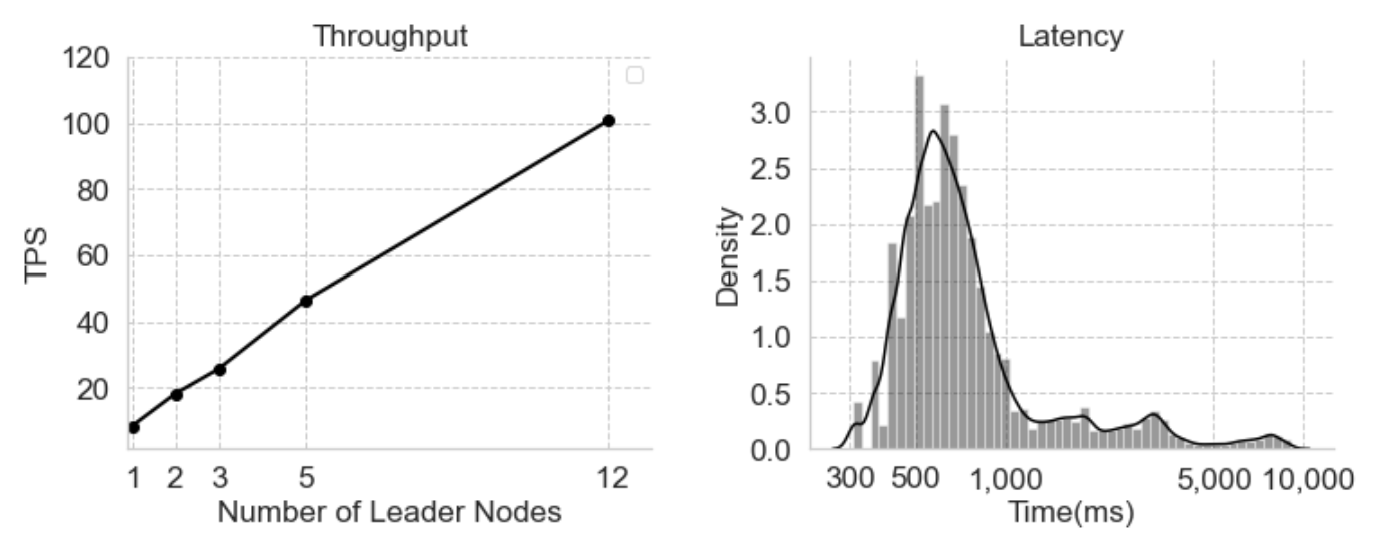} 
	\caption{\textbf{The result on the left figure shows a linearly increasing TPS (Transaction per second), which presents excellent user scalability and network scalability. The right figure shows an acceptable level of latency for most of the transactions.}\label{fig:Performance}}
\end{figure}

Sharding improves network scalability and user scalability. Commercial institutions, including tier-1 and tier-1.5 ones, could become leader nodes or validator nodes and undertake customer due diligence in retail networks. Sharding provides horizontal scalability to a CBDC network. 

If CBDC is non-fungible, token-based sharding can map every non-fungible token to one leader node when the issuer creates the token. Then tokens can be circulated efficiently in different ledgers in parallel, increasing the performance. However, CBDC is more like a fungible token, and every transaction produces a change. As a result, the token owner has to use several small tokens, which causes additional concurrent transactions. Moreover, if the token owner pays two tokens that circulate in different ledgers simultaneously, it must first initiate a cross-shard transaction and follow a single-shard transaction. The recommended consensus algorithm use sharding to improve performance by parallel running ledgers because cross-shard transactions are less frequent in the token-based sharding method than the account-based one. However, more shards may cause worse performance if each transaction in the retail consensus network needs verification from all parties in the wholesale networks. Therefore, the recommended algorithm makes one tier-1 institution decide transactions by itself, in which more shards do not bring worse performance. 

Overall, we prove that the proposed algorithms can increase the system's TPS while not sacrificing latency. 

\subsubsection{Privacy}
We provide two updated operating architectures to protect business secrecy. For option two, we designate one operating organization to distribute CBDC because it is not a competitor with other operating institutions such that their data are safe from monopoly.

For option one, the dynamic virtual addresses can prevent tier-1 institutions from knowing tier-1.5 institutions' customer identity. The method is similar to the bitcoin schema. In the bitcoin \cite{nakamoto2012bitcoin} system, essential facts exist: 1) transactions generate new addresses to collect change; 2) users could have many addresses. Bitcoin uses this method to protect customer privacy from leaking to the public, while our model protects tier-1.5 institutions' data from leaking to tier-1 institutions.

However, like the bitcoin schema\cite{bitcoinprivacy}, tier-1 institutions can still obtain secret information, depending on transaction types. We proved potential information leaking for different transactions types in Appendix A-B. Besides SISO transactions, SIDO, MIDO, MISO transactions may expose relationships between inputs and outputs. 

Therefore, a more aggressive method to protect user privacy is to create virtual entities. In the CBDC operating architecture, tier-1.5 institutions process transaction data before sending it to tier-1 institutions. tier-1.5 institutions can create virtual entities with virtual addresses in the network and use these virtual entities to create SISO transactions for customers. For example, in a SIDO transaction, tier-1.5 institutions can use virtual entities as the receiver to avoid the connection between payer and payee and then send it to the actual receiver via a SISO transaction. Other types of transactions can also apply virtual entities. Enough virtual entities can help tier-1.5 institutions to hide the direct relationship between the payees and payers.

\subsubsection{Security}
We followed the built model and proved that double-spending is possible in single-shard transactions and impossible in cross-shard transactions in Appendix A.

By our proof, we can get that a non-faulty leader prevents double-spending. However, the assumption is the weakest point in the system. The central bank usually does not perform malicious behaviour. So in the wholesale consensus network, the central bank can ensure no fault. However, in the retail consensus network, tier-1 institutions are responsible for their ledger and decide each transaction on its own without validation. No mechanism ensures them non-fault. As a result, double-spending may happen in single-shard transactions. 

The design of CBDC is a trade-off\cite{groupCBDC} between different CBDC technical features, including performance, security and privacy. In this example, one institution decides all transactions by itself, ensuring a high performance but sacrificing security. 

Although we can not avoid double-spending in real-time in the recommended consensus algorithms, we can increase the cost of malicious behaviours. For example, the recommended consensus algorithm in the retail consensus network chooses data backup from validators and the leader sends encrypted transactions to validators. The validators will undertake data backup. Once the leader node changes the original data, we can use the same encryption method to encrypt data and check data consistency. If the leader node performs malicious behaviour, they will be punished. Therefore, the leader nodes are reluctant to do malicious behaviours because the validators can find them easily. Besides, encryption can prevent the validators from accessing customer data, ensuring business secrecy in the CBDC system. In some cases, third-party auditors can run validator nodes, and it is not necessary to encrypt the data. 

On the other side, we can ensure extra verification will not overly influence latency because validators in the network are dynamic, so not all validators need to join the consensus process. Moreover, since the central bank controls issuance and redemption transactions, it knows the balance of money on each ledger so that no extra money comes from retail networks.

\subsection{Framework Iterations}
\begin{figure*}
	\centering
	\includegraphics[width=1.00\textwidth]{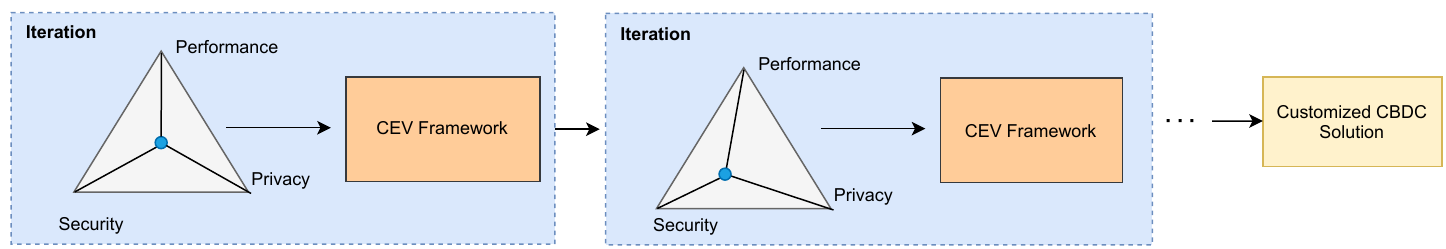}
	\caption{\textbf{Iterations of the CEV Framework until an acceptable balance point} \label{fig:iteration}}
\end{figure*} 
Figure \ref{fig:iteration} shows how we iterated the CEV framework. CBDC design involves many trade-offs\cite{groupCBDC} between different CBDC technical features. We start from a country with a large population, focusing on performance and privacy. Then we use the evaluation sub-framework to propose solutions. Finally, we leverage the verification framework to measure performance, privacy, and security. The proposed solution presents an excellent performance and privacy, but double-spending is possible. 

After verification, CBDC designers can return to CBDC technical features to adjust expectations to improve system security. For example, if they need a more secure system, they could continually use the evaluation sub-framework to propose new solutions and use the verification sub-framework to verify them. In the experiment, if CBDC designers want a real-time check for fault, they can make validators vote for each transaction. Then double-spending will not happen even though the leader node is faulty. The newly proposed solution can also involve more participants voting for single-shard transactions, ensuring no fault in the retail consensus networks. However, it may potentially influence the system's performance. Afterwards, the CBDC designer can return to CBDC technical features again until finding a balance point. The framework presents potential trade-offs in CBDC designs and helps CBDC designers find a balanced solution and meet expectations. 

\section{Conclusion}
Our paper proposes a CBDC framework (CEV Framework), including an evaluation sub-framework and a verification sub-framework to design solutions for CBDC systems. Our work proposes an original approach and potentially promotes the evolution of CBDC. Furthermore, to the best of our knowledge, we are the first to propose a framework to provide a holistic solution for CBDC designers according to their jurisdictions' economic and regulatory conditions. 

Our paper analyzes CBDC technical solutions by splitting consensus algorithms into different components and improving operating architectures to solve CBDC related issues. Most importantly, we build a verification sub-framework to prove the feasibility of the recommended algorithms and operating architectures with rigorous and professional empirical experiments and formal proofs. By using the CEV framework, diverse central bank digital currency projects can better design the consensus algorithms and adopt reasonable operating architectures. 

The framework could be continuously updated and improved by iterating with the workflow. The main future work is to include more CBDC technical features and solutions into the framework and update the impact table to make it more accurate to propose solutions. Besides, the CEV framework needs a more efficient way to verify proposed solutions. Additionally, the CEV framework can also be used to propose stablecoin solutions.

\section*{Acknowledgment}
This article is partially based on their work in the Global CBDC Challenge \cite{b1}, in which they were shortlisted as finalists. Based on the framework in this paper, they built an evaluation and recommendation platform that advocates consensus algorithms and operating architectures for different
CBDC designers. This platform can satisfy different national economic and regulatory conditions. The authors wish to acknowledge the other two teammates, Mark Liu and Bing Qu. They also gratefully acknowledge the help of Bo Tong Xu in fruitful discussions. This paper is reviewed by Dr. Philip
Intallura, Dr. Bing Zhu, and Dr. Ziyuan Li. They also wish to express great appreciation for their valuable input.

\section*{Appendix A Formal Logic Proof}
Here are some temporal logic proofs in the paper. Please see notations in Section IV-B.1.
\subsection{Double-spending Proof}
The premises below come from lemmas or definitions in Section IV. The logic proofs \cite{temporal logic} have been checked by the proof-editor from Stanford University\cite{standford}.
 \begin{lemma}
 A non-issuance transaction makes the out-degrees of the input tokens become non-zero in the ledger state.
 $$ \forall x. \forall \tau.(Tx(\tau,x) \Rightarrow \forall y.(xRy \Rightarrow d_{G}^{+}(\tau,y)!=0).$$
  \end{lemma}
 \begin{proof}
 Since  ($\forall$x. $\forall$ $\tau$.(Tx($\tau$,x) $\Rightarrow$ D = D + TD)), the ledger operator adds a new transaction graph into the ledger where the out-degrees of input tokens keep same. A successful non-issuance transaction in definition 3 makes the out-degrees of input tokens become non-zero.
 \end{proof}
 
 \begin{lemma}
$\forall$x. $\forall$ $\tau$.(Tx($\tau$,x) $\Rightarrow$ $\forall$y.(xRy $\Rightarrow$ H($\tau$, y)))
  \end{lemma}
\begin{proof}
Lemma 1 shows $\forall$x. $\forall$ $\tau.$(Tx($\tau$,x) $\Rightarrow$ $\forall$y.(xRy $\Rightarrow$  $d_{G}^{+}(\tau,y)!=0$). 

\footnotesize
\begin{logicproof}{1}
 $$\forall x. \forall \tau.(Tx(\tau,x) \Rightarrow \forall y.(xRy \Rightarrow  d_{G}^{+}(\tau,y)!= 0$$)) & Premise\\
 \forall y.\forall\tau.(H(\tau, y) $$\Leftrightarrow$$ d_{G}^{+}(\tau,y)!$$=$$0) & Premise\\
 Tx(\tau,x) $$\Rightarrow$$ \forall y.(xRy $$\Rightarrow$$ d_{G}^{+}(\tau,y)!$$=$$0) & UE: 1\\
 H(\tau, y) $$\Leftrightarrow$$ d_{G}^{+}(\tau,y)!$$=$$0 & UE: 2 \\
 Tx(\tau,x) & AS \\
 \forall y.(xRy $$\Rightarrow$$ d_{G}^{+}(\tau,y)!$$=$$0) & IE: 3,5\\
 xRy $$\Rightarrow$$ d_{G}^{+}(\tau,y)!$$=$$0 & UE: 6\\
  \begin{subproof}
 xRy & AS \\
 d_{G}^{+}(\tau,y)!$$=$$0 & IE: 7,8\\
 d_{G}^{+}(\tau,y)!$$=$$0 $$\Rightarrow$$ H(\tau, y) & BE: 4\\
 H(\tau, y) & IE: 9,10
   \end{subproof}
 xRy $$\Rightarrow$$ H(\tau,y) & II: 8,11\\

 \forall y.(xRy $$\Rightarrow$$  H(\tau,y)) & UI: 12 \\

 Tx(\tau,x) $$\Rightarrow$$\forall y.(xRy $$\Rightarrow$$  H(\tau,y)) & II: 5,13\\
 \forall x. \forall \tau.(Tx(\tau,x) $$\Rightarrow$$\forall y.(xRy $$\Rightarrow$$ H(\tau,y))) & UI: 14
\end{logicproof}
\end{proof}

 \begin{lemma}
 $$ \forall x.\forall \tau.(\mathrm{Tx}(\tau,x)) \Rightarrow \forall z.(zRx \Rightarrow F(\tau, z))) $$
 \end{lemma}
\begin{proof}
In the model, one token $\tau$ keeps unspent status when it has not been involved in any transaction. For any $ \tau \in V, F(\tau, y) \Leftrightarrow (\tau \in$ UTXO at time y) $\Leftrightarrow d_{G}^{+}(\tau,y)=0$. From the definition in transaction, we get
$\forall$x. $\forall$ $\tau$.(Tx($\tau$,x) $\Rightarrow$ $\forall$y.(yRx $\Rightarrow$  $d_{G}^{+}(\tau,y)=0$).

\footnotesize
\begin{logicproof}{1}
 $$\forall x. \forall \tau.(Tx(\tau,x) \Rightarrow \forall y.(yRx \Rightarrow  d_{G}^{+}(\tau,y)= 0$$)) & Premise\\
 \forall y.\forall\tau.(F(\tau, y) $$\Leftrightarrow$$ d_{G}^{+}(\tau,y)$$=$$0 )& Premise\\
 Tx(\tau,x) $$\Rightarrow$$ \forall y.(yRx $$\Rightarrow$$ d_{G}^{+}(\tau,y)$$=$$0) & UE: 1\\
 F(\tau, y) $$\Leftrightarrow$$ d_{G}^{+}(\tau,y)$$=$$0 & UE: 2 \\
 Tx(\tau,x) & AS \\
 \forall y.(yRx $$\Rightarrow$$ d_{G}^{+}(\tau,y)$$=$$0) & IE: 3,5\\
 yRx $$\Rightarrow$$ d_{G}^{+}(\tau,y)$$=$$0 & UE: 6\\
 \begin{subproof}
 yRx & AS \\
 d_{G}^{+}(\tau,y)$$=$$0 & IE: 7,8\\
 d_{G}^{+}(\tau,y)$$=$$0 $$\Rightarrow$$ F(\tau, y) & BE: 4\\
  F(\tau, y) & IE: 9,10 
    \end{subproof}
  yRx $$\Rightarrow$$ F(\tau,y) & II: 8,11\\
  \forall y.(yRx $$\Rightarrow$$ F(\tau, y) ) & UI: 12\\
  Tx(\tau,x) $$\Rightarrow$$\forall y.(yRx $$\Rightarrow$$  F(\tau, y)) & II: 5,13\\
  \forall x. \forall \tau.(Tx(\tau,x) $$\Rightarrow$$\forall y.(yRx $$\Rightarrow$$  F(\tau, y))) & UI: 14
 
\end{logicproof}
\end{proof}

 \begin{lemma}
 $$ \forall x.\forall \tau.(F(\tau,x) \Leftrightarrow \neg H(\tau,x)) $$
 \end{lemma}
\begin{proof}
$\neg H( \tau, x)$ means $\tau$ has not been spent before x. Therefore, F($\tau$,x) $\Leftrightarrow$  $d_{G}^{+}(\tau,x)!=0$ $\Leftrightarrow$ $\neg H( \tau, x)$.

\footnotesize
\begin{logicproof}{1}
 \forall x. \forall \tau.(F(\tau,x)$$\Leftrightarrow$$ d_{G}^{+}(\tau,x)!$$=$$0 ) & Premise\\
  \forall x. \forall \tau.(H(\tau,x)$$\Leftrightarrow$$ d_{G}^{+}(\tau,x)!$$=$$0 ) & Premise\\
\forall \tau.(F(\tau,x)$$\Leftrightarrow$$ d_{G}^{+}(\tau,x)!$$=$$0 ) & UE:1 \\
F(\tau,x)$$\Leftrightarrow$$ d_{G}^{+}(\tau,x)!$$=$$0  & UE:3 \\
\forall \tau.(H(\tau,x)$$\Leftrightarrow$$ d_{G}^{+}(\tau,x)!$$=$$0 ) & UE:2 \\
H(\tau,x)$$\Leftrightarrow$$ d_{G}^{+}(\tau,x)!$$=$$0  & UE:5 \\
F(\tau,x)$$\Rightarrow$$ d_{G}^{+}(\tau,x)!$$=$$0  & BE:4 \\
d_{G}^{+}(\tau,x)!$$=$$0 $$\Rightarrow$$ F(\tau,x)   & BE:4 \\
H(\tau,x)$$\Rightarrow$$ d_{G}^{+}(\tau,x)!$$=$$0  & BE:6 \\
d_{G}^{+}(\tau,x)!$$=$$0 $$\Rightarrow$$ H(\tau,x)   & BE:6 \\
 \begin{subproof}
F(\tau,x)   & AS \\
d_{G}^{+}(\tau,x)!$$=$$0 & IE:7,11\\
H(\tau,x)  & IE: 10, 12 
\end{subproof}
F(\tau,x)$$\Rightarrow$$H(\tau,x) & II:11,13 \\
\begin{subproof}
H(\tau,x)  & AS\\
d_{G}^{+}(\tau,x)!$$=$$0 & IE:9,15\\
F(\tau,x)   & IE:8,16 
\end{subproof}
H(\tau,x)$$\Rightarrow$$ F(\tau,x) & II:15,17 \\
F(\tau,x)$$\Leftrightarrow$$ H(\tau,x) & BI:14,18 \\
\forall \tau.(F(\tau,x)$$\Leftrightarrow$$ H(\tau,x)) & UI:19 \\
\forall x.\forall \tau.(F(\tau,x)$$\Leftrightarrow$$ H(\tau,x)) & UI:20 
\end{logicproof}
\end{proof}

\begin{lemma}

 $$\forall x. \forall y.(Tx(\tau_{1},x) \land Tx(\tau_{2},y) \land x R y \Rightarrow \tau_{1} \neq \tau_{2})$$
 \end{lemma}
\begin{proof}
Time is continuous that one timestamp always exists between any two timestamps. we assume a double spending transaction possible as one premise and find a contradiction.

\footnotesize
\begin{logicproof}{1}
\forall x. \forall \tau.(\mathrm{Tx}(\tau,x) \Rightarrow \forall y.(x R y $$\Rightarrow$$ H(\tau, y))) & Premise \\
\forall x. \forall \tau.(\mathrm{Tx}(\tau,x)) \Rightarrow \forall z.(z R x $$\Rightarrow$$ F(\tau, z))) & Premise \\
 \forall x.\forall \tau.(F(\tau,x) $$\Leftrightarrow$$ \neg H( \tau,x)) & Premise\\
 \forall x\forall y.(xRy $$\Rightarrow$$ (\exists z. (xRz $$\land$$zRy))) & Premise\\
 \exists x.\exists y. (xRy$$\land$$Tx(\tau,x), Tx(\tau,y)) & goal\\
 \exists y.([x]Ry $$\land$$Tx(\tau,[x]), Tx(\tau,y)) & EE: 5\\
 [x]R[y] $$\land$$Tx(\tau,[x]), Tx(\tau,[y]) & EE: 6\\
 [x]R[y] & AE: 7\\
Tx(\tau,[x]) & AE: 7\\
Tx(\tau,[y]) & AE: 7\\
\forall y.([x]Ry $$\Rightarrow$$ (\exists z. ([x]Rz $$\land$$zRy)) & UE: 4\\
[x]R[y] $$\Rightarrow$$ (\exists z. ([x]Rz $$\land$$ zR[y])) & UE: 11\\
\exists z. ([x]Rz $$\land$$zR[y]) & IE: 8,12\\
[x]R[z] $$\land$$[z]R[y] & EE: 13 \\
[x]R[z] & AE: 14 \\
[z]R[y] & AE: 14 \\
 \forall \tau.(\mathrm{Tx}(\tau,[x]) \Rightarrow \forall y.([x] R y$$ \Rightarrow$$ H(\tau, y))) & UE: 1\\
\forall \tau.(\mathrm{Tx}(\tau,[y]) \Rightarrow \forall z.(z R [y] \Rightarrow F(\tau, z))) & UE: 2\\
\mathrm{Tx}(\tau,[x]) $$\Rightarrow$$ \forall y.([x] R y $$\Rightarrow$$ H(\tau, y)) & UE: 17\\
\mathrm{Tx}(\tau,[y]) $$\Rightarrow$$ \forall z.(z R [y] $$\Rightarrow$$ F(\tau, z)) & UE: 18\\
 \forall y.([x] R y$$ \Rightarrow$$ H(\tau, y)) & IE: 9,19 \\
 \forall z.(z R [y] $$\Rightarrow$$ F(\tau, z)) & IE: 10,20 \\
 [x] R [z] $$\Rightarrow$$ H(\tau, [z]) & UE: 21 \\
 [z] R [y] $$\Rightarrow$$ F(\tau, [z]) & UE: 22 \\
 H(\tau, [z])) & IE: 15,23\\
 F(\tau, [z])) & IE: 16,24\\
 \forall \tau.(F(\tau,[z]) $$\Leftrightarrow$$ \neg H( \tau,[z]) & UE: 3 \\
 F(\tau,[z]) $$\Leftrightarrow$$ \neg H( \tau,[z]) & UE: 27 \\
 F(\tau,[z]) $$\Rightarrow$$ \neg H( \tau,[z]) & BE: 28 \\ 
 \neg H( \tau,[z]) & IE: 26,29 \\
 Contradiction & 25,30
\end{logicproof}
With proof by contradiction, we get that a recorded token in the validator's ledger can not be spent twice in different transactions.
\end{proof}

\begin{lemma}
Assume a leader is non-faulty(H), double-spending will not happen in its shard.\end{lemma}
\begin{proof}
 In a token chain, the issuer creates and issues tokens (a) with the in-degree of 0. For a valid payment transaction, a non-faulty leader ensures itself record the received token and change token in a valid transaction.
 
\footnotesize
\begin{logicproof}{1}
 p(a) & Premise
 \\
 \forall \tau. \, (H $$\land$$p(\tau) $$\Rightarrow$$ p(r(\tau)) $$\land$$ p(c(\tau))) & Premise\\
 H $$\land$$p(\tau) $$\Rightarrow$$ p(r(\tau)) $$\land$$p(c(\tau)) & UE: 2 \\
 \begin{subproof}
 H & AS \\
  p(\tau) & AS \\
  H $$\land$$p(\tau) & AI: 4,5\\
 p(r(\tau)) $$\land$$ p(c(\tau)) & IE: 3,6 \\
 p(r(\tau)) & AE: 7\\ 
 p(\tau) $$\Rightarrow$$ p(r(\tau)) & II: 5,8\\
 \forall \tau. \, ( p(\tau)$$ \Rightarrow$$ p(r(\tau)) & UI: 9
 \end{subproof}
 
 H $$\Rightarrow$$ (\forall \tau. \, (p(\tau)$$\Rightarrow$$ p(r(\tau)))) & II: 4,10\\
 \begin{subproof}
 H & AS \\
  p(\tau) & AS \\
  H $$\land$$p(\tau) & AI: 12,13\\
 p(r(\tau)) $$\land$$p(c(\tau)) & IE: 3,12 \\
 p(c(\tau)) & AE: 15\\ 
 p(\tau) $$\Rightarrow $$p(c(\tau)) & II: 13,16\\
 \forall \tau. \, ( p(\tau) $$\Rightarrow $$p(c(\tau)) & UI: 17 
 \end{subproof}
 H $$\Rightarrow$$ (\forall \tau. \, (p(\tau)$$ \Rightarrow$$ p(c(\tau)))) & II: 12,18 \\
 \begin{subproof}
 H & AS \\
 \forall \tau. \, (p(\tau) $$\Rightarrow$$ p(r(\tau))) &IE: 11,20 \\
 \forall \tau. \, (p(\tau) $$\Rightarrow $$p(c(\tau))) &IE: 19,20\\
 \forall \tau. \, (p(\tau)) &	Ind: 1, 21, 22
 \end{subproof}
 H $$\Rightarrow$$ (\forall \tau. \, p(\tau)) &	II: 20,23
\end{logicproof}
\end{proof}

\begin{lemma}
If leaders are non-faulty(H), double-spending will not happen in cross-shard transactions.
\end{lemma}
\begin{proof}
The model shows that a non-faulty leader with the central bank ensures that the leader records new tokens in the ledger. Lemma 5 proves a recorded token has no double-spending problems. Since the leader node records the tokens in the ledger successfully, no double-spending exist in a cross-shard transaction. In our model, the currency issuer (Central Bank) secures the Initial Issuance Transaction and the Final Redemption Transaction and makes sure the token recorded.
\end{proof}

\begin{lemma}
If leaders are non-faulty(H), double-spending will not happen in all transactions.
\end{lemma}
\begin{proof}
There are two kinds of transactions in the network: single-shard transactions and cross-shard transactions. Lemma 6 and 7 prove no double-spending with non-faulty leaders in these two kinds of transactions. Therefore, we conclude no double-spending problem in the network if leaders are non-faulty.
\end{proof}

\subsection{Privacy Proof}
\begin{figure}[h]
	\centering
	\includegraphics[width=0.50\textwidth]{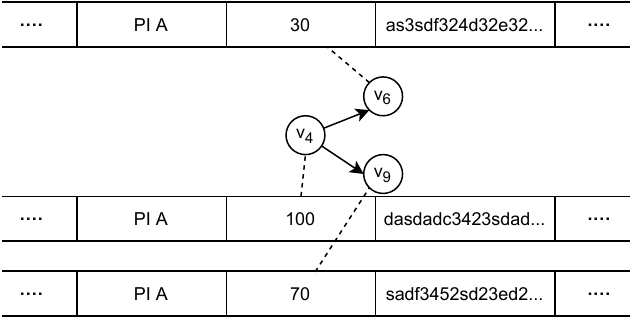}
	\caption{\textbf{One of $v_{6}$ and $v_{9}$ is the change money back to the owner of $v_{4}$. The relationship between payees and payers could be inferred when collecting enough extra data, like goods, transaction places.} \label{fig:privacyproof}}
\end{figure} 

\begin{lemma}
SISO transaction protects the identity relationship between payers and payees the most.
\end{lemma}
\begin{proof}
Here are all types of transactions and their privacy protection capability. 
\begin{enumerate}
    \item SIDO transaction: figure \ref{fig:privacyproof} shows a SIDO transaction, in which one of the output tokens should be the change token back to the payer.
    \item MISO transaction: the payers pays several tokens from different virtual addresses. Then these virtual addresses possibly belong to one person.
    \item MIDO transaction: the payers are usually the same person.
    \item SISO transaction: the payees and payers usually are not the same people. 
\end{enumerate}
Overall, we concluded that SISO transaction can protect identity privacy.
\end{proof}

\end{document}